\documentclass[12pt]{amsart}
\usepackage{amssymb,mathrsfs,fullpage,mathtools,upref}
\usepackage{anyfontsize,newtxtext,newtxmath,hyperref}

\newtheorem{theorem}{Theorem}[section]

\newtheorem{cor}[theorem]{Corollary}
\newtheorem*{cor*}{Corollary}
\newtheorem{lemma}[theorem]{Lemma}

\theoremstyle{definition}
\newtheorem{definition}[theorem]{Definition}

\theoremstyle{remark}
\newtheorem{remark}[theorem]{Remark}

\let\originalleft\left
\let\originalright\right
\renewcommand{\left}{\mathopen{}\mathclose\bgroup\originalleft}
\renewcommand{\right}{\aftergroup\egroup\originalright}

\numberwithin{equation}{section}

\DeclareMathOperator{\tr}{Tr}

\newcommand{\bb}[1]{\mathbb{#1}}
\newcommand{\cl}[1]{\mathcal{#1}}
\newcommand{\ff}[1]{\mathfrak{#1}}
\newcommand{\bh}{\mathfrak{B}(\mathcal{H})}

\newcommand{\inner}[2]{\left\langle {#1},{#2} \right\rangle}
\newcommand{\norm}[1]{\left\| #1 \right\|}
\newcommand{\abs}[1]{\left| #1 \right|}

\begin{document}
\title[Constant-sized robust self-tests for states and measurements of unbounded dimension]{Constant-sized robust self-tests for states and measurements of unbounded dimension}

\author[L. Man\v{c}inska]{Laura Man\v{c}inska}
\address{Department of Mathematical Sciences, Universitetsparken 5, 2100 K{\o}benhavn, Denmark}
\email{mancinska@math.ku.dk}

\author[J. Prakash]{Jitendra Prakash}
\address{Department of Mathematical Sciences, Universitetsparken 5, 2100 K{\o}benhavn, Denmark}
\email{jp@math.ku.dk}

\author[C. Schafhauser]{Christopher Schafhauser}
\address{Department of Mathematics, University of Nebraska -- Lincoln, USA}
\email{cschafhauser2@unl.edu}

\date{\today}

\begin{abstract} 
We consider correlations, $p_{n,x}$, arising from measuring a maximally entangled state using $n$ measurements with two outcomes each, constructed from $n$ projections that add up to $xI$. We show that the correlations $p_{n,x}$ robustly self-test the underlying states and measurements. To achieve this, we lift the group-theoretic Gowers-Hatami based approach for proving robust self-tests to a more natural algebraic framework. A key step is to obtain an analogue of the Gowers-Hatami theorem allowing to perturb an ``approximate" representation of the relevant algebra to an exact one.

For $n=4$, the correlations $p_{n,x}$ self-test the maximally entangled state of every odd dimension as well as 2-outcome projective measurements of arbitrarily high rank. The only other family of constant-sized self-tests for strategies of unbounded dimension is due to Fu (QIP 2020) who presents such self-tests for an infinite family of maximally entangled states with \emph{even} local dimension. Therefore, we are the first to exhibit a constant-sized self-test for measurements of unbounded dimension as well as all maximally entangled states with odd local dimension. 
\end{abstract}

\maketitle
\tableofcontents

\section{Introduction}
One of the key tasks in the development of reliable quantum technologies is the certification of quantum devices. This ensures that the devices are performing according to their specification. One of the ways such certification may be carried out is by using \emph{self-testing} methods which enable us to infer the quantum-mechanical description of a device merely from classical observations (measurement statistics). We can then treat these devices as black boxes as we need not trust the inner workings of the system, a scenario which one refers to as \emph{device-independence}. Beginning with~\cite{MY04}, in which the term was first coined, self-testing has found many applications, such as {device-independent quantum cryptography}~\cite{MY98, MY04}, {delegated quantum computation}~\cite{CGJV}, {entanglement detection}~\cite{Bow18a, Bow18b}, investigating the structure of the {quantum correlation set}~\cite{CS17,Goh18}, and {quantum complexity theory}~\cite{Fitz-et-al1,NV,NW}. Self-testing is also one of the ingredients behind the recent breakthrough result establishing that $\text{MIP}^* = \text{RE}$~\cite{JNVWY}; which further implies a negative answer to the celebrated {Connes' embedding problem} \cite{Connes} from the theory of von Neumann algebras.

In this work, we prove that certain quantum strategies constructed from projections which sum up to some particular scalar times identity can be robustly self-tested from the quantum correlations that they induce. Specifically, we begin with $d\times d$ projections $\widetilde{P}_1,\dots,\widetilde{P}_n$ ($n\geq 3$) such that $\widetilde{P}_1+\dots+\widetilde{P}_n = x I_d$ for  some specific scalar $x\in \bb R$. The scalar $x$ and the projections $\widetilde{P}_1,\dots,\widetilde{P}_n$ have the property that whenever $P_1,\dots,P_n$ are any other projections such that $P_1 + \dots + P_n = xI$, then $P_i = I\otimes \widetilde{P}_i$ for all $1\leq i\leq n$ in some basis. We then define a quantum strategy $\widetilde{\mathscr S}$ and its induced quantum correlation $\widetilde{p}$ corresponding to these projections. We show that if one observes the quantum correlation $\widetilde{p}$ induced from an arbitrary quantum strategy $\mathscr S$, then $\mathscr S$ is related to $\widetilde{\mathscr S}$ via a local isometry. For the robust case, we prove that if we observe a quantum correlation $p$ such that $\|p-\widetilde{p}\|\leq \epsilon$, then its inducing strategy $\mathscr S$ is ``approximately'' related to $\widetilde{\mathscr S}$ via a local isometry. 

Often times self-testing results are proven using ad-hoc techniques, which means that to obtain new results one essentially has to start from scratch. One notable exception is an approach that
uses perturbative representation theory of groups to establish robust self-testing~\cite{Vidick,CS17,Cui2020}. The basic idea is that in the ideal case one deduces algebraic relations which the measurement operators must satisfy on the quantum state. One then associates a suitable finite group such that one gets a representation of the group. This transfers to the well-studied field of the representation theory of finite groups and the problem reduces to identifying suitable irreducible representations of the group. In the robust case, one gets approximate versions of the algebraic relations which yield ``approximate'' representations of the group (with respect to state-dependent distance). A key tool used in the approximate case is the Gowers--Hatami theorem~\cite{GH17,Gowers17,Vidick} which relates an ``approximate'' representation of a group to a representation of that group in some suitable sense. 
A caveat to this approach is that, in general, it is not clear what group should be associated to the algebraic relations identified and if this is at all possible. In fact, it would be more natural to associate an \emph{algebra} rather than a group to these algebraic relations. Moreover, there are known families of non-local games whose optimal strategies do not have an apparent underlying group yet an underlying algebra can easily be identified ({\it e.g.} binary constraint system games \cite{clevemittal} beyond linear ones, synchronous games \cite{PSSTW, HMPS}, graph homomorphism \cite{qhomos} and isomorphism games \cite{qiso1}).

One of our main contributions is that we showcase how the above general method for self-testing from this \emph{group} framework can be lifted to an \emph{algebraic} framework. Instead of seeking an appropriate group to associate with the algebraic relations, we simply work with the algebra generated by those relations. To accomplish this, a major step is to obtain some sort of analogue of Gowers--Hatami theorem for algebras. We show how this can be done for a particular algebras, but the approach can be easily generalized to algebras arising from other synchronous correlations whenever a version of Gowers-Hatami Theorem holds; we expect this to be the  case for algebras with well-understood representation theory such as finite dimensional algebras. The analogue we prove in this paper has a non-constructive ($\epsilon$,$\delta$)-dependence which is the most pressing question left open by the current work.

In comparison to measurements, self-testing of quantum states is relatively well understood. For example, we know that any (pure) bipartite entangled state in $\bb C^d\otimes \bb C^d$ can be self-tested from a correlation with $3$ inputs and $d$ outputs~\cite{Coladangelo2017}. For applications, it would be efficient to have small-sized correlations that robustly self-test states with large dimensions. The only family of \emph{constant-sized} correlations that self-test states of arbitrarily large dimension \cite{Fu19} are constructed by building upon Slofstra's group embedding procedure \cite{Slofstra19} for linear binary constraint system games. Specifically \cite{Fu19} shows that for each $d\in \mathscr{D}$, where $\mathscr{D}$ is an infinite subset of the primes given by \begin{align*}
\mathscr{D} = \lbrace d: d \text{ is an odd prime, and the smallest generator of the group } \bb Z_d^{\times} \text{ is } 2, 3 \text{ or } 5 \rbrace,
\end{align*} the maximally entangled state $\varphi_{4(d-1)}$ can be robustly self-tested from correlations with over $100$ questions per party. We obtain the following corollary which complements the result in~\cite{Fu19}, but with correlations of considerably smaller size:

\begin{cor*}
For each odd dimension $d\geq 3$, the maximally entangled state $\varphi_d$ can be robustly self-tested by quantum correlations with four inputs and two outputs.
\end{cor*}

When compared to \cite{Fu19}, the strength of our proof is the simplicity and the small correlation size while the weakness is the non-constructive $(\epsilon,\delta)$ dependence in the robustness proof. 

In their work \cite{Fu19}, the author establishes certain algebraic relationships between the measurements needed to induce the considered constant-sized correlations but they do not present\footnote{Nevertheless, we believe that a continuation of the arguments presented in \cite{Fu19} should lead to constant-sized self-tests of measurements.} a self-test of measurements according to the standard definition \cite{SB19}. Therefore, our work is the first one to show that measurements of arbitrarily large dimension can be self-tested from constant-sized correlations. In addition, to the best of our knowledge, we are the first to establish self-testing of measurements with operators of rank higher than one:

\begin{cor*}
Given any natural number $k$ there exist four projections of rank $k$ which can be robustly self-tested by quantum correlations with four inputs and two outputs.
\end{cor*}

Most of the known self-tests for \emph{infinite families} of measurements are for tensor-products of Pauli matrices (for example,~\cite{NV17,Col17}) or Clifford unitaries~\cite{CGJV}. There are a few results which are different from these, for instance, \cite{Sarkar,Cui2020}. Our main theorem yields another example of an infinite family of measurements that goes beyond a tensor-product of Paulis or Clifford unitaries.

The following example illustrates the concept of robust self-testing and also outlines the method that we use to obtain our results. Consider the \emph{CHSH game} \cite{CHSH,CHTW}, where two non-communicating players, Alice and Bob, want to win the game against a referee. The referee picks a pair of bits $(v,w)\in \{0,1\}^2$ uniformly and sends $v$ to Alice and $w$ to Bob. Alice replies with a bit $i$ and Bob replies with a bit $j$. Alice and Bob win the game if $vw = i+j\pmod 2$. Classical strategies can only help them win the game with probability 3/4.  On the other hand, quantum strategies provide a better winning probability of $\omega_{q}(\mathrm{CHSH}) \approx 0.85$. Indeed, the following quantum strategy (which we term \emph{canonical}) \begin{align}\label{eq:canon-strat-CHSH}
	\widetilde{\mathscr S} = \left(\varphi_2\in \bb C^2\otimes \bb C^2, \left\lbrace \widetilde{A}_0 = Z, \widetilde A_1 = X\right\rbrace, \left\lbrace \widetilde B_0 = \frac{Z+X}{\sqrt{2}},\widetilde B_1 = \frac{Z-X}{\sqrt{2}}\right\rbrace \right),
\end{align} where $\varphi_2$ is the maximally entangled state and $X,Z$ are the Pauli matrices, achieves the maximum quantum winning probability.

Suppose that Alice and Bob use an unknown quantum strategy (in terms of observables) given by \begin{align}\label{eq:strat-S}
	\mathscr{S} = \left(\psi\in \bb C^{d_A}\otimes \bb C^{d_B}, \left\lbrace A_0,A_1\right\rbrace,\left\lbrace B_0,B_1\right\rbrace\right)
\end{align} which also yields the maximum winning probability $\omega_{q}(\mathrm{CHSH})$. What can then be said about the strategy $\mathscr S$? It has been shown that in such a case, there exist isometries $V_A\colon \bb C^{d_A}\to \bb C^2\otimes \cl K_A$ and $V_B\colon \bb C^{d_B}\to \bb C^2\otimes \cl K_B$ for some Hilbert spaces $\cl K_A, \cl K_B$, and a quantum state $\psi_{\mathrm{junk}}\in \cl K_A\otimes \cl K_B$ such that \begin{align}
	(V_A\otimes V_B)\psi &= \varphi_2\otimes \psi_{\mathrm{junk}}, \text{ and} \label{eq:eq1}\\
	(V_A\otimes V_B)(A_i \otimes B_j)\psi &= ((\widetilde{A}_i \otimes \widetilde{B}_j)\varphi_2) \otimes \psi_{\mathrm{junk}}, \text{ for all } i,j\in \{0,1\}. \label{eq:eq2} \end{align} That is, all quantum strategies which yield the maximum quantum winning probability for the CHSH game are related to the canonical one via local isometries, which can be viewed as a rigidity property of $\widetilde{\mathscr S}$.

This rigidity result can be extended to the approximate case (which we term robust self-testing): say the winning probability $\omega_{q}(\mathscr S)$ obtained by the strategy $\mathscr{S}$ in \eqref{eq:strat-S} above satisfies $|\omega_{q}(\mathscr{S}) - \omega_{q}(\mathrm{CHSH})|\leq \epsilon,$ for some $\epsilon \geq 0$. Then, in this case, there exist isometries $V_A\colon \bb C^{d_A}\to \bb C^2\otimes \cl K_A$ and $V_B\colon \bb C^{d_B}\to \bb C^2\otimes \cl K_B$ for some Hilbert spaces $\cl K_A, \cl K_B$, and a quantum state $\psi_{\mathrm{junk}}\in \cl K_A\otimes \cl K_B$ which satisfy approximate versions of Equations \eqref{eq:eq1} and \eqref{eq:eq2}: \begin{align}
	\|(V_A\otimes V_B)\psi - \varphi_2\otimes \psi_{\mathrm{junk}}\| &\leq f(\epsilon), \text{ and} \label{eq:eq3}\\
	\|(V_A\otimes V_B)(A_i \otimes B_j)\psi - ((\widetilde{A}_i \otimes \widetilde{B}_j)\varphi_2) \otimes \psi_{\mathrm{junk}}\| &\leq f(\epsilon), \text{ for all } i,j \in \{0,1\}, \label{eq:eq4}
\end{align} for some $f(\epsilon) \geq 0$ such that $f(\epsilon)\to 0$ as $\epsilon \to 0$. 

There are many proofs for these results \cite{MY04,MYS12,Kan17}. We describe the approach taken in \cite{Vidick} as this illustrates the connection between self-testing and representation theory of groups. One starts by noting that the canonical observables $\widetilde{A}_0,\widetilde{A}_1$ for Alice in \eqref{eq:canon-strat-CHSH} generate a finite group $\cl P$ (called the \emph{Pauli group} which is isomorphic to the dihedral group of order eight). In the case when $\omega_{q}(\mathscr{S}) = \omega_{q}(\mathrm{CHSH})$, the function $f: \cl P \to \mathscr U_{d_A}$ (where $\mathscr U_{d_A}$ is the group of $d_A\times d_A$ unitaries) defined by $\widetilde{A}_0\mapsto A_0$ and $\widetilde{A}_1\mapsto A_1$ defines a group representation with respect to the state $\rho_A = \tr_B(\psi\psi^*)$. On the other hand, in the approximate case $|\omega_{q}(\mathscr{S}) - \omega_{q}(\mathrm{CHSH})|\leq \epsilon$, one gets an ``approximate'' group representation with respect to the state $\rho_A$. In the second step, to get the existence of isometries one invokes the Gowers--Hatami Theorem \cite{GH17, Vidick}.

\begin{theorem}[Gowers--Hatami]\label{thm:GH}
	Let $G$ be a finite group, $d\in \bb N$, $\epsilon>0$, $\sigma\in \bb M_d$ be a density matrix, and $f:G\to \mathscr{U}_d$ be an $(\epsilon,\sigma)$-representation of $G$. Then, there exists $D\geq d$, an isometry $V\colon \bb C^d \to \bb C^{D}$, and a representation $g\colon G\to \mathscr{U}_{D}$ such that $\norm{f(a) - V^*g(a)V}_{\sigma} \leq |G|\sqrt{\epsilon}$ for all $a\in G$.
\end{theorem}

Finally, these two steps are stitched together to get the robust self-testing statements for the CHSH game.

For the proof of our results, we follow the blueprint mentioned above for the CHSH game with certain necessary modifications.  As noted above, optimal strategies for the CHSH game may be described in terms of unitary representations of the Pauli group $\cl P$.  Further, unitary representations of $\cl P$ are in bijection with Hilbert space representations of the group C$^*$-algebra $\mathrm{C}^*(\cl P) \cong \mathbb C^4 \oplus \mathbb M_2$. In our situation, the strategies we consider do not arise from unitary representations of a group but will still arise from Hilbert space representations of a certain C$^*$-algebra $\mathscr P_{n, x}$ (which is not isomorphic to a group C$^*$-algebra). The C$^*$-algebra $\mathscr P_{n, x}$ is the universal C$^*$-algebra generated by projections $r_1,\dots,r_n$ satisfying the algebraic relation $r_1+\dots+r_n=x1$.  Studying the (approximate) representation theory of this C$^*$-algebra will provide an analogue of Theorem \ref{thm:GH} for $\mathscr P_{n, x}$, which leads to the desired self-testing results in Section~\ref{self-testing-sum-proj}.  Analogues of Theorem \ref{thm:GH} are ubiquitous in C$^*$-algebra theory, but usually approximations are considered in operator norm (c.f. \cite{Loring}) or in the 2-norm arising from a tracial state (c.f. \cite{Hadwin-Shulman}).  In our case, we will need to work with the 2-norm arising from an ``approximately tracial'' state which presents new difficulties (see Theorem \ref{thm:Gowers-Hatami-N}).

\subsection{Organisation of the paper} We fix some notations and conventions, and some introductory material in Section~\ref{sec:prelim}. In Section~\ref{sec:q-syn}, we define quantum strategies, their induced quantum correlations, and the special subset of synchronous quantum correlations. We prove some of its properties that we will be using later. We give a formal definition of robust self-testing in Section~\ref{sec-self-test-def}. In Section~\ref{sec:proj-sum-scalars}, we collect some of the basic properties of projections adding up to scalar times identity, and we define (Definition \ref{def:def-of-p*}) the families of quantum strategies obtained by such projections and the correlations that they induce. Finally, in Section \ref{self-testing-sum-proj}, we show that the quantum correlations in Definition \ref{def:def-of-p*} robustly self-test their canonical quantum strategies, and in Section \ref{sec:implications} we discuss some of its implications.  We include two short appendices: Appendix \ref{appendixA} on some key definitions from C$^*$-algebras that we'll be using in the article, and Appendix \ref{appendixB} to outline a method to construct four projections adding to a scalar.

\section{Preliminaries}\label{sec:prelim}
Let $\bb N$ denote the set of natural numbers $\left\lbrace 1,2,3,\dots\right\rbrace$. Given $n\in \bb N$, define $\left[n\right]\coloneqq\left\lbrace 1,\dots,n\right\rbrace$. A fraction $\frac{b}{d}$ ($b,d\in \bb N$) is said to be in \emph{lowest terms} if $\gcd(b,d) = 1$. The cardinality of a set $X$ is denoted $|X|$.

For each $d\in \bb N$ we let $\bb R^d$ and $\bb C^d$ represent the $d$-dimensional real and complex Euclidean spaces, respectively. For $d_A,d_B\in \bb N$ we let $\bb M_{d_A,d_B}$ denote the space of all $d_A\times d_B$ complex matrices. In particular, for $d\in \bb N$, we let $\bb M_d \coloneqq \bb M_{d,d}$ denote the algebra of all $d\times d$ complex matrices. The $d\times d$ identity matrix is denoted by $I_d$. For $X\in \bb M_{d_A,d_B}$ we let $X^T$ denote the transpose of $X$. For $X \in \bb M_d$ we let $\tr(X)$ denote its trace, whereas we shall use $\mathrm{tr}_d(X) \coloneqq \frac{1}{d}\tr(X)$ to denote its \emph{normalized} trace.

Let $\cl H$ be a (complex) Hilbert space.  We assume that the inner product on $\cl H$ is linear in the first argument and conjugate-linear in the second. We denote the algebra of all bounded operators on $\cl H$ by $\bh$. The identity operator on $\cl H$ will be denoted by $I_{\cl H}$.  Given vectors $\xi,\eta\in \cl H$ and an $\epsilon \geq 0$, we write $\xi\approx_{\epsilon}\eta$ if $\|\xi-\eta\| \leq \epsilon$. We will usually identify a $d$-dimensional Hilbert space $\cl H$ with $\bb C^d$. The algebra of all operators on $\bb C^d$ is identified with $\bb M_d$ with respect to the standard orthonormal basis $\{e_i\}_{i=1}^d$ of $\bb C^d$. Given $X\in \bb M_d$, a subspace $\cl K\subseteq \bb C^d$ is called an \emph{invariant subspace} of $X$ if $X(\cl K) \subseteq \cl K$.

A matrix $X\in \bb M_d$ is called \emph{positive} if $\inner{X\xi}{\xi}\geq 0$ for all $\xi\in \bb C^d$. For two Hermitian matrices $X,Y\in \bb M_d$ we say that $X\leq Y$ if $Y-X$ is positive. A matrix $P\in \bb M_d$ is called a \emph{projection} if $P=P^*=P^2$, where $P^*$ is the adjoint of $P$. A finite set of positive matrices $X_1,\dots,X_k \in \bb M_d$ is called a \emph{positive operator-valued measure} (POVM) if $X_1+\dots+X_k = I_d$. In particular, a finite set of projections $P_1,\dots,P_k \in \bb M_d$ is called a \emph{projection-valued measure} (PVM) if $P_1+\dots+P_k = I_d$.

We recall the Schmidt decomposition \cite[Section~5.3]{Peres}.

\begin{lemma}[Schmidt decomposition]
Let $\psi\in \bb C^{d_A}\otimes \bb C^{d_B}$. There exist orthonormal sets $\{\xi_l\}_{l=1}^r$ in $\bb C^{d_A}$, $\{\eta_l\}_{l=1}^r$ in $\bb C^{d_B}$, and strictly positive real numbers $\{\alpha_l\}_{l=1}^r$ such that $\psi = \sum_{l=1}^r \alpha_l \xi_l\otimes \eta_l$. The numbers $\{\alpha_l\}_{l=1}^r$, called the \emph{Schmidt coefficients} of $\psi$, are unique up to permutation.
\end{lemma}

Let $\psi = \sum_{l=1}^r \alpha_l \xi_l\otimes \eta_l$ be a Schmidt decomposition of a vector $\psi\in \bb C^{d_A}\otimes \bb C^{d_B}$. Any orthonormal basis of $\bb C^{d_A}$ (resp., $\bb C^{d_B}$) containing $\{\xi_l\}_{l=1}^r$ (resp., $\{\eta_l\}_{l=1}^r$) is called a \emph{Schmidt basis} for $\bb C^{d_A}$ (resp., $\bb C^{d_B}$). The number $r$ is called the \emph{Schmidt rank} of $\psi$, and we say that $\psi$ is of \emph{full Schmidt rank} if $r = d_A = d_B$. Finally, define the subspaces $\mathrm{supp}_{A}(\psi) \coloneqq \mathrm{span}\{\xi_l\}_{l=1}^r$ and $\mathrm{supp}_B(\psi) \coloneqq \mathrm{span}\{\eta_l\}_{l=1}^r$.

A matrix $\rho \in \bb M_d$ is called a \emph{density matrix} if $\rho\geq 0$ and $\tr(\rho) = 1$. If $\rho\in \bb M_{d_A}\otimes \bb M_{d_B}$ is a density matrix, then the density matrices \begin{align*}
\rho_A \coloneqq \tr_B(\rho) \coloneqq (\mathrm{id}_A \otimes \tr)(\rho), \quad \text{ and } \quad \rho_B \coloneqq \tr_A(\rho) \coloneqq (\tr \otimes \mathrm{id}_B)(\rho),
\end{align*} where $\mathrm{id}_A$ and $\mathrm{id}_B$ are identity maps on $\bb M_{d_A}$ and $\bb M_{d_B}$, respectively, are called its \emph{reduced density matrices}. 

A density matrix $\rho\in \bb M_d$ induces a semi-inner product on $\bb M_d$ by \begin{align*}
\inner{X}{Y}_{\rho} \coloneqq \tr_{\rho}(Y^*X) \coloneqq \tr(Y^*X\rho), \qquad \text{ for all } X,Y\in \bb M_d.
\end{align*} This semi-inner product induces a semi-norm $\|X\|_{\rho} \coloneqq  \inner{X}{X}_{\rho}^{1/2}$ for all $X\in \bb M_d$. This semi-inner product (resp., semi-norm) is an inner product (resp., a norm) if and only if $\rho$ is invertible.

Usually we will be in the following situation. 	Let $\psi\in \bb C^{d_A}\otimes \bb C^{d_B}$ be a unit vector. Let $\rho = \psi\psi^*$ be the corresponding density matrix and let $\rho_A, \rho_B$ be its reduced density matrices. Then, $\inner{.}{.}_{\rho_A}$ and $\inner{.}{.}_{\rho_B}$ are semi-inner products on $\bb M_{d_A}$ and $\bb M_{d_B}$, respectively, and they are inner products if and only if $\psi$ is of full Schmidt rank. Observe that for any $X\in \bb M_{d_A}$ and $Y\in \bb M_{d_B}$, \begin{align}
\norm{X}_{\rho_A}^2 &= \tr(X^*X\rho_A) = \inner{\left(X^*X\otimes I_{d_B}\right)\psi}{\psi} = \|(X\otimes I_{d_B})\psi\|^2, \text{ and} \label{partial-trace-1} \\ 
\norm{Y}_{\rho_B}^2 &= \tr(Y^*Y\rho_B) = \inner{\left(I_{d_A}\otimes Y^*Y\right)\psi}{\psi} = \|(I_{d_A}\otimes Y)\psi\|^2 \label{partial-trace-2}.
\end{align}

The \emph{operator-vector correspondence} is the linear map $\mathrm{vec} \colon \bb M_{d_A,d_B} \rightarrow \bb C^{d_A}\otimes \bb C^{d_B}$ given on the canonical matrix units $\{E_{i,j}:(i,j)\in [d_A]\times [d_B]\}$ for $\mathbb M_{d_A, d_B}$ by $\mathrm{vec}(E_{i,j}) \coloneqq e_i\otimes e_j$ for all $(i,j)\in [d_A]\times [d_B]$. In particular, for $\xi\in \bb C^{d_A}$ and $\eta\in \bb C^{d_B}$, we have $\mathrm{vec}(\xi\eta^*) = \xi\otimes \overline{\eta}$. The map $\mathrm{vec}$ is an isometry: $\inner{\mathrm{vec}(X)}{\mathrm{vec}(Y)} = \inner{X}{Y}_2 \coloneqq \tr(Y^*X)$, for $X,Y\in \bb M_{d_A,d_B}$. We shall use the following identity: for $X\in \bb M_{d_A',d_A}, Y\in \bb M_{d_B',d_B}, D\in \bb M_{d_A,d_B}$, \begin{align}\label{eq:vec-identity}
(X\otimes Y)\mathrm{vec}(D) &= \mathrm{vec}(XDY^T).
\end{align} For each $d\in \bb N$, the unit vector $\varphi_d = \frac{1}{\sqrt{d}}\sum_{i=1}^de_i\otimes e_i$ is called the (\emph{canonical}) \emph{maximally entangled state} on $\bb C^d\otimes \bb C^d$. Then the diagonal matrix $D=\frac{1}{\sqrt{d}}I_d$ satisfies $\mathrm{vec}(D)=\varphi_d$. For $X,Y\in \bb M_d$, using Identity \eqref{eq:vec-identity} and the isometry property, we have \begin{align}
(X\otimes Y)\varphi_d & = \frac{1}{\sqrt{d}}\mathrm{vec}(XY^T), \text{ and} \label{eq:vec-iden-31} \\
\inner{(X\otimes Y)\varphi_d}{\varphi_d} &= \frac{1}{d}\tr(XY^T). \label{eq:vec-iden-32}
\end{align}

\section{Synchronous quantum correlations}\label{sec:q-syn}
Let us briefly describe the framework of non-local games \cite{CHTW} in which the discussion of correlations and strategies takes a natural form. A \emph{non-local game} $\mathscr{G}$ is played by two spatially separated players, Alice and Bob, against a referee. There are (non-empty) finite sets $I_A$ and $I_B$, called the \emph{question sets}, and $O_A$ and $O_B$, called the \emph{answer sets}, and a \emph{rule function} 
\[ \lambda \colon I_A\times I_B\times O_A\times O_B\to \{0,1\}, \] 
all of which are known to the players. However, once the game begins, the players cannot communicate with each other. The game begins with the referee randomly picking a pair of questions $(v,w)\in I_A\times I_B$ according to some probability distribution\footnote{The probability distribution on $I_A \times I_B$ should be given as part of the game, but as we are only interested in winning (or perfect) strategies, the probability distribution will not be relevant for us.  For our purposes, it suffices to view $I_A \times I_B$ with the uniform distribution, so each pair of questions is chosen with equal probability.} and sending $v$ to Alice and $w$ to Bob. Alice returns an answer $i\in O_A$ to the referee, and so does Bob with $j\in O_B$. Alice and Bob win the round of the game if $\lambda(v,w,i,j) = 1$, or otherwise they lose.  The goal of Alice and Bob is to maximize their winning probability, and for this, they can take the help of either classical resources or quantum resources. It is well-known that for many games, quantum resources exhibit larger winning probabilities as compared to classical resources \cite{CHTW}. A classic example is the \emph{CHSH game} as described in the Introduction; here $I_A=I_B=O_A=O_B=\{0,1\}$, and $\lambda(v,w,i,j)=1$ if and only if $vw = i+j\pmod 2$.

As the game is played for many rounds, an outside observer will see a probability distribution which encodes the strategies employed by the players.  Let $p(i,j|v,w)$ be the probability of Alice responding with $i$ and Bob responding with $j$ given that they received $v$ and $w$, respectively. \emph{Classical strategies} usually employ either a \emph{deterministic} strategy (given by a function $I_A \times I_B \rightarrow O_A \times O_B$) or by using \emph{local/shared randomness}. On the other hand, \emph{quantum strategies} are given by a \emph{quantum state} (usually \emph{entangled}) shared by Alice and Bob and sets of \emph{quantum measurements}. Specifically, we have the following definition.

\begin{definition}[Quantum strategies and quantum correlations]
Given (non-empty) finite sets $I_A,I_B,O_A,O_B$ with $|I_A|=n_A, |I_B|=n_B, |O_A|=k_A$ and $|O_B|=k_B$, a \emph{quantum strategy} $\mathscr{S}$ is given by a triple: \begin{align}
\mathscr{S} = \left(\psi\in \bb C^{d_A}\otimes \bb C^{d_B}, \left\lbrace E_{v,i}: i\in O_A, v\in I_A \right\rbrace, \left\lbrace F_{w,j}: j\in O_B, w\in I_B \right\rbrace \right),
\end{align} where $\psi$ is a unit vector (also referred to as a \emph{quantum state}), for each $v\in I_A$ the set $\{E_{v,i}:i\in O_A\}$ is a POVM on $\bb C^{d_A}$, and for each $w\in I_B$, the set $\{F_{w,j}:j\in O_B\}$ is a POVM on $\bb C^{d_B}$.

Each quantum strategy $\mathscr{S}$ induces a collection of probability distributions called a \emph{quantum correlation} as follows: for each $(v,w)\in I_A\times I_B$, define \begin{align}\label{eq:def-q-corr}
p(i,j|v,w) = \inner{\left(E_{v,i}\otimes F_{w,j}\right)\psi}{\psi}, \quad \text{ for all } (i,j)\in O_A\times O_B.
\end{align}

Fixing $n_A,n_B,k_A,k_B\in \bb N$, the set of quantum correlations arising from all choices of quantum strategies is denoted by $C_{q}(n_A,n_B,k_A,k_B)$. In particular, we define $C_q(n,k) \coloneqq C_q(n,n,k,k)$ for all $n,k\in \bb N$.
\end{definition}

If $p = (p(i,j|v,w))\in C_{q}(n_A,n_B,k_A,k_B)$ is a quantum correlation, then the POVM assumptions imply that $p(i,j|v,w)\geq 0$ for all $i,j,v$ and $w$, and 
\begin{align*}
\sum_{(i,j)\in O_A\times O_B} p(i,j|v,w) = 1, \quad \text{for all } (v,w)\in I_A\times I_B.
\end{align*} Moreover, quantum correlations satisfy \emph{non-signaling conditions}; that is, they have well-defined \emph{marginals} $p_A(i|v)$ and $p_B(j|w)$ in the sense that 
\begin{align}
p_A(i|v) &\coloneqq \sum_{j' \in O_B} p(i,j'|v,w) = \sum_{j' \in O_B} p(i,j'|v,w'), \text{ and}\label{eq:non-sig-q1} \\
p_B(j|w) &\coloneqq \sum_{i' \in O_A} p(i',j|v,w) = \sum_{i' \in O_A} p(i',j|v',w), \label{eq:non-sig-q2}
\end{align} for all $v,v'\in I_A, w,w'\in I_B, i\in O_A, j\in O_B$. Indeed, using the definition of a POVM, Equations \eqref{eq:non-sig-q1} and \eqref{eq:non-sig-q2} yield 
\begin{align}
p_A(i|v) = \inner{\left(E_{v,i}\otimes I_{d_B}\right)\psi}{\psi} \quad \text{and} \quad p_B(j|w) = \inner{\left(I_{d_A} \otimes F_{w,j}\right)\psi}{\psi}.
\end{align} Thus $p_A(i|v)$ is independent of $w$, and $p_B(j|w)$ is independent of $v$. (In the context of non-local games, the non-signaling condition is imposed due to the assumption that the players may not communicate.)

\begin{remark}
The set $C_q(n_A,n_B,k_A,k_B)$ of quantum correlations is convex for all natural numbers $n_A,n_B,k_A$ and $k_B$. However, in general, it is not closed. This was first shown by Slofstra who proved that $C_q(184,235,8,2)$ is not closed \cite{Slofstra19}. Thereafter, non-closure of quantum correlations sets with smaller values of $n_A,n_B,k_A,k_B$ have been proved \cite{DPP,MR20,C20}---in particular, $C_q(5, 2)$ is not closed as was first shown in \cite{DPP}.
\end{remark}

\begin{remark}
For notational simplicity, we shall simply take 
\begin{align*}
I_A = [n_A], I_B = [n_B], O_A = [k_A], \text{ and } O_B = [k_B]
\end{align*} 
in the definition of a quantum correlation.
\end{remark}

Introduced in \cite{PSSTW}, the subset of synchronous correlations arises naturally in the context of many non-local games, for example, the \emph{graph coloring games} \cite{CHTW} and \emph{graph homomorphism games} \cite{MR14}. In this article, we shall be self-testing some quantum strategies yielding synchronous quantum correlations, and therefore here is a formal definition.

\begin{definition}\label{def:synCq}
A correlation $(p(i,j|v,w))\in C_{q}(n,k)$ is called \emph{synchronous} if $k \coloneqq k_A = k_B$, $n \coloneqq n_A = n_B$, and $p(i,j|v,v)=0$ for all $i,j \in [k]$ and $v\in [n]$ with $i \neq j$. We let $C_q^s(n,k)$ denote the subset of all synchronous correlations of $C_q(n,k)$, and call it the set of \emph{synchronous quantum correlations}.
\end{definition}

A quantum strategy which induces a synchronous quantum correlation has several nice properties. For instance, when restricted to $\mathrm{supp}_A(\psi)$ (resp., $\mathrm{supp}_B(\psi)$), the measurement operators for Alice (resp., Bob) are projections, Bob's measurement operators can be expressed in terms of Alice's, and $X\mapsto \inner{(X\otimes I)\psi}{\psi}$ defines a \emph{tracial state} (see Appendix \ref{appendixA}) on the C$^*$-algebra generated by Alice's measurements \cite[Theorem 5.5]{PSSTW}. The next two lemmas show that if a correlation is ``approximately'' synchronous, then these properties ``approximately'' hold for its inducing strategy. These lemmas are not new and appear in various places in the literature in different guises, but for the sake of completeness, we provide proofs.

Since correlations are collections of probability distributions, it makes sense to work with the $1$-norm. That is, given two correlations $p_1,p_2\in C_{q}(n_A,n_B,k_A,k_B)$, we define \begin{align}\label{eq:1-norm}
\norm{p_1-p_2}_1 \coloneqq \sum_{i,j,v,w}\abs{p_1(i,j|v,w) - p_2(i,j|v,w)}.
\end{align}

\begin{lemma}\label{lem:approx-vectors} Let $p\in C_q(n,k)$ be induced by a strategy \begin{align*}
\mathscr{S} = \left(\psi\in \bb C^{d_A}\otimes \bb C^{d_B}, \left\lbrace E_{v,i}: i\in [k], v\in [n] \right\rbrace, \left\lbrace F_{w,j}: j\in [k], w\in [n] \right\rbrace \right).
\end{align*} Let $\rho_A, \rho_B$ be the reduced density matrices of the density matrix $\rho = \psi\psi^*$. If $\|p-\widetilde{p}\|_1 \leq \delta$ for some synchronous $\widetilde{p}\in C_q^s(n,k)$ and for some $\delta\geq 0$, then for all $i\in [k]$ and $v \in [n]$, \begin{enumerate}
\item[(a)] $\|(E_{v,i}\otimes I_{d_B})\psi - (I_{d_A}\otimes F_{v,i})\psi\| \leq \sqrt{\delta}$,
\item[(b)] $\|(E_{v,i}\otimes I_{d_B})\psi - (E_{v,i}\otimes F_{v,i})\psi\| \leq \sqrt{\delta}$,
\item[(c)] $\|(I_{d_A}\otimes F_{v,i})\psi - (E_{v,i}\otimes F_{v,i})\psi\| \leq \sqrt{\delta}$,
\item[(d)] $\|E_{v,i} - E_{v,i}^2\|_{\rho_A} =  \|(E_{v,i}\otimes I_{d_B})\psi - (E_{v,i}^2\otimes I_{d_B})\psi\| \leq 2\sqrt{\delta}$,
\item[(e)] $\|F_{v,i} - F_{v,i}^2\|_{\rho_B} =  \|(I_{d_A}\otimes F_{v,i})\psi - (I_{d_A} \otimes F_{v,i}^2)\psi\| \leq 2\sqrt{\delta}$.
\end{enumerate}\end{lemma}

\begin{proof}
Using the fact that $0\leq E_{v,i}^2 \leq E_{v,i}$ and $0\leq F_{v,i}^2 \leq F_{v,i}$ for all $v\in [n]$,  \begin{align*}
\|(E_{v,i}\otimes I_{d_B})\psi - (I_{d_A}\otimes F_{v,i})\psi\|^2 &= \inner{(E_{v,i}^2\otimes I_{d_B})\psi}{\psi} + \inner{(I_{d_A}\otimes F_{v,i}^2)\psi}{\psi}  - 2 \inner{(E_{v,i}\otimes F_{v,i})\psi}{\psi} \\
&\leq \inner{(E_{v,i}\otimes I_{d_B})\psi}{\psi} + \inner{(I_{d_A}\otimes F_{v,i})\psi}{\psi}  - 2 \inner{(E_{v,i}\otimes F_{v,i})\psi}{\psi} \\
&= p_A(i|v) + p_B(i|v) - 2 p(i,i|v,v) \\
&= \sum_{\substack{j=1 \\ j\neq i}}^k p(i,j|v,v) + p(j,i|v,v) \\
&= \sum_{\substack{j=1 \\ j\neq i}}^k |p(i,j|v,v) - \widetilde{p}(i,j|v,v)| + |p(j,i|v,v) - \widetilde{p}(j,i|v,v)| \\
&\leq \|p-\widetilde{p}\|_1 \leq \delta,
\end{align*} where the penultimate line is due to the fact that $\widetilde{p}$ is synchronous. Thus (a) follows. 

Similarly, to see (b), note that \begin{align*}
\|(E_{v,i}\otimes I_{d_B})\psi - (E_{v,i}\otimes F_v)\psi\|^2 &= \| (E_{v,i} \otimes (I_{d_B} - F_{v,i})) \psi \|^2 \\
&= \inner{(E_{v,i}^2 \otimes (I_{d_B} - F_{v,i})^2) \psi}{\psi} \\
&\leq \inner{(E_{v,i} \otimes (I_{d_B} - F_{v,i}))\psi}{\psi} \\
&= \inner{(E_{v,i} \otimes I_{d_B})\psi}{\psi} - \inner{(E_{v,i} \otimes  F_{v,i})\psi}{\psi} \\
&= p_A(i|v) -  p(i,i|v,v) \\
&= \sum_{\substack{j=1 \\ j\neq i}}^k p(i,j|v,v) = \sum_{\substack{j=1 \\ j\neq i}}^k |p(i,j|v,v) - \widetilde{p}(i,j|v,v)| \\
&\leq \|p - \widetilde{p}\|_1 \leq \delta.
\end{align*} Part (c) is similar.

The equalities in parts (d) and (e) follow from Equations \eqref{partial-trace-1} and \eqref{partial-trace-2}.
Since $\|E_{v,i}\|\leq 1$ and $\|F_{v,i}\|\leq 1$, we deduce from (a) that for all $v\in [n]$,
\begin{align*}
\|(E_{v,i}^2 \otimes I_{d_B})\psi - (E_{v,i}\otimes F_{v,i})\psi\| &\leq \sqrt{\delta} \text{ and}\\
\|(I_{d_A}\otimes F_{v,i}^2)\psi - (E_{v,i}\otimes F_{v,i})\psi\| &\leq \sqrt{\delta}.
\end{align*} 
Using these two inequalities with (b) and (c) and applying triangle inequality, we get the inequalities in part (d) and (e). 
\end{proof}

We record the special case of $\delta = 0$ in Lemma \ref{lem:approx-vectors}.

\begin{cor}\label{cor:proj-strat}
Let $p\in C_q^s(n,k)$ be a synchronous correlation induced by a strategy \begin{align*}
\mathscr{S} = \left(\psi\in \bb C^{d_A}\otimes \bb C^{d_B}, \left\lbrace E_{v,i}: i\in [k], v\in [n] \right\rbrace, \left\lbrace F_{w,j}: j\in [k], w\in [n] \right\rbrace \right).
\end{align*} Then, for all $i\in [k]$ and $v\in [n]$, \begin{enumerate}
\item[(a)] $(E_{v,i}\otimes I_{d_B})\psi = (I_{d_A}\otimes F_{v,i})\psi$,
\item[(b)] $(E_{v,i}\otimes I_{d_B})\psi = (E_{v,i}^2\otimes I_{d_B})\psi$,
\item[(c)] $(I_{d_A}\otimes F_{v,i})\psi = (I_{d_A}\otimes F_{v,i}^2)\psi$.
\end{enumerate}
Moreover, if $\psi$ is of full Schmidt rank $r = d_A = d_B$, then for all $v\in [n]$, the matrices $E_{v,i}, F_{v,i}$ are projections, and $F_{v,i} = D^TE_{v,i}^T(D^{-1})^T$, where $D \in \bb M_r$ is such that $\mathrm{vec}(D) = \psi$.
\end{cor}

\begin{proof}
Statements (a), (b), and (c) follow immediately from the corresponding statements in Lemma \ref{lem:approx-vectors}.  When $\psi$ has full Schmidt rank, the seminorms $\|\cdot\|_{\rho_A}$ and $\|\cdot\|_{\rho_B}$ are norms, and so $E_{v,i}$ and $F_{v,i}$ are projections by statements (d) and (e) in Lemma \ref{lem:approx-vectors}.  Finally, (a) and Equation \eqref{eq:vec-identity} imply $E_{v,i} D = D F_{v,i}^T$.  Since $\psi$ has full Schmidt rank, $D$ is invertible, and the last equation follows.
\end{proof}

\begin{lemma}[Approximately tracial state]\label{lem:st-app-tra}
Let $p\in C_q(n,k)$ be induced by a strategy \begin{align*}
\mathscr{S} = \left(\psi\in \bb C^{d_A}\otimes \bb C^{d_B}, \left\lbrace E_{v,i}: i\in [k], v\in [n] \right\rbrace, \left\lbrace F_{w,j}: j\in [k], w\in [n] \right\rbrace \right).
\end{align*} Let $\rho_A,\rho_B$ be the reduced density matrices of the density matrix $\rho = \psi\psi^*$.  Let $\cl A\subseteq \bb M_{d_A}$ (resp., $\cl B\subseteq \bb M_{d_B}$) be the $C^*$-algebra generated by $\{E_{v,i}\}_{v,i}$ (resp., $\{F_{w,j}\}_{w,j}$). If $\|p-\widetilde{p}\|_1 \leq \delta$ for some synchronous $\widetilde{p}\in C_q^s(n,k)$ and for some $\delta\geq 0$, then the states $\tr_{\rho_A}\colon X\mapsto \tr(X\rho_A)$ on $\cl A$ and $\tr_{\rho_B}\colon Y\mapsto \tr(Y\rho_B)$ on $\cl B$ satisfy 
\begin{align}
|\tr_{\rho_A}(WX - XW)| &\leq 2\ell\sqrt{\delta} \text{ and} \label{approx-tracial1} \\
|\tr_{\rho_B}(UY - YU)| &\leq 2\ell\sqrt{\delta}, \label{approx-tracial2}
\end{align} respectively, where $W$ (resp., $U$) is any word in $\{E_{v,i}\}_{v,i}$ (resp., $\{F_{w,j}\}_{w,j}$) of length $\ell$, $X\in \cl A$ with $\|X\|\leq 1$, and $Y\in \cl B$ with $\|Y\|\leq 1$.
\end{lemma}

\begin{proof}
Let $W = E_{v,i}W'$ where $W'$ is of length $\ell-1$, $i\in [k]$ and $v\in [n]$. Then making use of Lemma \ref{lem:approx-vectors}(a) a couple of times and that $\|E_{v,i}\|\leq 1$ and $\|X\|\leq 1$, we have \begin{align*}
\tr_{\rho_A}(WX) &= \inner{(E_{v,i}W'X\otimes I_{d_B})\psi}{\psi} \\
&= \inner{(W'X\otimes I_{d_B})\psi}{(E_{v,i}\otimes I_{d_B})\psi} \\
&\approx_{\sqrt{\delta}} \inner{(W'X\otimes I_{d_B})\psi}{(I_{d_A}\otimes F_{v,i})\psi} \\
&= \inner{(W'X\otimes F_{v,i})\psi}{\psi} \\
&= \inner{(W'X\otimes I_{d_B})(I_{d_A}\otimes F_{v,i})\psi}{\psi} \\
&\approx_{\sqrt{\delta}} \inner{(W'X\otimes I_{d_B})(E_{v,i}\otimes I_{d_B})\psi}{\psi} \\
&= \inner{(W'XE_{v,i}\otimes I_{d_B})\psi}{\psi}.
\end{align*} Thus, continuing in this fashion for the whole length of $W'$, we will obtain Expression \eqref{approx-tracial1}.  Similarly, Expression \eqref{approx-tracial2} holds.
\end{proof}

\begin{remark}
In Lemma \ref{lem:approx-vectors} and Lemma \ref{lem:st-app-tra}, it is enough to assume that $p\in C_q(n,k)$ satisfies $\abs{p(i, j|v, v)} \leq \delta$ for all $i \neq j$ and for all $v\in [n]$.
\end{remark}

\section{Robust self-testing: definition and basic properties}\label{sec-self-test-def}
In many cases, given a quantum correlation $\widetilde p$, there is a ``unique'' quantum strategy $\widetilde{\mathscr{S}}$ which induces the correlation $\widetilde p$, in the sense that, any other strategy $\mathscr S$ which also induces the correlation $\widetilde p$ is related to $\widetilde{\mathscr{S}}$ by a \emph{local isometry}. The way these two strategies $\mathscr{S}$ and $\widetilde{\mathscr{S}}$ are related is what we term as a \emph{local dilation}. For our purpose, we work with a more general definition of a \emph{local $\epsilon$-dilation} which relates strategy $\mathscr{S}$ to $\widetilde{\mathscr{S}}$ in the case when the correlation $p$ induced by $\mathscr{S}$ is $\epsilon$-close to $\widetilde{p}$ (in the 1-norm as in Equation \ref{eq:1-norm}). This general definition is useful in practical scenarios where it is inevitable that errors creep in while conducting experiments. 

\begin{definition}[Local $\epsilon$-dilation]\label{def:loc-eps-dil}
Given $\epsilon \geq 0$ and two finite-dimensional strategies \begin{align*}
\widetilde{\mathscr{S}} &= \big(\widetilde{\psi}\in \widetilde{\cl H}_A\otimes \widetilde{\cl H}_B, \{\widetilde{E}_{v,i}:v\in [n_A], i\in [k_A]\},\{\widetilde{F}_{w,j}:w\in [n_B], j\in [k_B] \}\big), \\
\mathscr{S} &= \left(\psi\in \cl H_A\otimes \cl H_B, \{E_{v,i}:v\in [n_A], i\in [k_A]\},\{F_{w,j}:w\in [n_B], j\in [k_B]\}\right),
\end{align*} we say that $\widetilde{\mathscr{S}}$ is a \emph{local $\epsilon$-dilation} of $\mathscr{S}$ if there exist isometries $V_A\colon \cl H_A \to \widetilde{\cl H}_A \otimes \cl K_A$ and $V_B\colon \cl H_B \to \widetilde{\cl H}_B \otimes \cl K_B$ for some finite-dimensional Hilbert spaces $\cl K_A$ and $\cl K_B$, and a quantum state $\psi_{\mathrm{junk}}\in \cl K_A\otimes \cl K_B$ such that \begin{align}
(V_A\otimes V_B)\psi &\approx_{\epsilon} \widetilde{\psi}\otimes \psi_{\mathrm{junk}} \text{ and} \label{eq:approx-loc-dil-1} \\
\left(V_A\otimes V_B\right)(E_{v,i}\otimes F_{w,j})\psi &\approx_{\epsilon} \big((\widetilde{E}_{v,i}\otimes \widetilde{F}_{w,j})\widetilde{\psi}\big) \otimes \psi_{\mathrm{junk}} \label{eq:approx-loc-dil-2}
\end{align} for all $i$, $j$, $v$, and $w$.

When $\epsilon = 0$, we say that $\widetilde{\mathscr{S}}$ is a \emph{local dilation} of $\mathscr{S}$ (instead of a local $0$-dilation). In this case, the approximations in Expressions \eqref{eq:approx-loc-dil-1} and \eqref{eq:approx-loc-dil-2} are replaced by equalities.
\end{definition}

\begin{remark}
We are abusing notation slightly in the above definition.  In Estimates \eqref{eq:approx-loc-dil-1} and \eqref{eq:approx-loc-dil-2}, the vectors on the left belong to the Hilbert space $(\widetilde{\cl H}_A \otimes \cl K_A) \otimes (\widetilde{\cl H}_B \otimes \cl K_B)$ and the vectors on the right belong to the Hilbert space $(\widetilde{\cl H}_A \otimes \widetilde{\cl H}_B) \otimes (\cl K_A \otimes \cl K_B)$.  We are identifying these two Hilbert spaces via the unitary which flips the second and third tensor factors.
\end{remark}

\begin{remark}\label{remark:terminology}
A note on the terminology that we have chosen: In the literature (for instance \cite{SB19}), the phrase ``equivalent up to a local isometry'' is often used instead of ``local dilation''. However, we feel that using the word ``equivalence'' is somewhat misleading, and that ``local dilation'' would be an appropriate replacement, since the term ``local dilation'' has a directional connotation (for instance, we can write $\mathscr{S}\xhookrightarrow{\epsilon} \widetilde{\mathscr{S}}$ to mean that $\widetilde{\mathscr{S}}$ is a local $\epsilon$-dilation of $\mathscr{S}$) which is not present if we use the term ``equivalence up to local isometry''.
\end{remark}

\begin{remark}
Observe that if a strategy $\widetilde{\mathscr{S}}$ induces a correlation $\widetilde{p}$, and if $\widetilde{\mathscr{S}}$ is a local dilation of any other strategy $\mathscr{S}$, then $\mathscr{S}$ also induces the same correlation.
\end{remark}

We now recall the definition of self-testing and robust self-testing.

\begin{definition}[Self-testing]\label{def:self-test}
A correlation $\widetilde{p} \in C_{q}(n_A,n_B,k_A,k_B)$ \emph{self-tests} a strategy $\widetilde{\mathscr{S}}$ if for any strategy $\mathscr{S}$ which also induces $\widetilde{p}$, $\widetilde{\mathscr{S}}$ is a local dilation of $\mathscr{S}$.
\end{definition}

\begin{definition}[Robust self-testing]\label{def:robust-self-test}
	A correlation $\widetilde{p} \in C_{q}(n_A,n_B,k_A,k_B)$ self-tests a strategy $\widetilde{\mathscr{S}}$ \emph{robustly} if $\widetilde{p}$ self-tests $\widetilde{\mathscr{S}}$ and the following condition is satisfied. For each $\epsilon \geq 0$, there exists $\delta \geq 0$, such that if there is some $p\in C_{q}(n_A,n_B,k_A,k_B)$ with $\norm{p-\widetilde{p}}_1 \leq \delta$, and $p$ is induced by a strategy $\mathscr{S}$, then $\widetilde{\mathscr{S}}$ is a local $\epsilon$-dilation of $\mathscr{S}$.
\end{definition}

For local $\epsilon$-dilations, we establish a transitivity statement which is helpful in reducing certain results to the case when the quantum state is of full Schmidt rank. In terms of the notation in Remark \ref{remark:terminology}, we establish  \begin{align*}
\big(\mathscr{S}_3\xhookrightarrow{\epsilon_2} \mathscr{S}_2 \xhookrightarrow{\epsilon_1} \mathscr{S}_1 \big) \Rightarrow \big(\mathscr{S}_3\xhookrightarrow{\epsilon_1+\epsilon_2} \mathscr{S}_1 \big).
\end{align*}

\begin{lemma}[Transitivity of local $\epsilon$-dilations]\label{lem:trans-epsilon-dilation}
Let $\mathscr{S}_1, \mathscr{S}_2$ and $\mathscr{S}_3$ be finite-dimensional quantum strategies.  Let $\epsilon_1,\epsilon_2 \geq 0$. If $\mathscr{S}_1$ is a local $\epsilon_1$-dilation of $\mathscr{S}_2$, and $\mathscr{S}_2$ is a local $\epsilon_2$-dilation of $\mathscr{S}_3$, then $\mathscr{S}_1$ is a local $(\epsilon_1+\epsilon_2)$-dilation of $\mathscr{S}_3$.
\end{lemma}

\begin{proof}
Let the quantum strategies be given by \begin{align*}
\mathscr{S}_1 &= (\psi_1\in \cl H_{A,1}\otimes \cl H_{B,1}, \{E_{v,i}^{(1)}:v\in [n_A], i\in [k_A]\},\{F_{w,j}^{(1)}:w\in [n_B], j\in [k_B]\}), \\
\mathscr{S}_2 &= (\psi_2\in \cl H_{A,2}\otimes \cl H_{B,2}, \{E_{v,i}^{(2)}:v\in [n_A], i\in [k_A]\},\{F_{w,j}^{(2)}:w\in [n_B], j\in [k_B]\}), \\
\mathscr{S}_3 &= (\psi_3\in \cl H_{A,3}\otimes \cl H_{B,3}, \{E_{v,i}^{(3)}:v\in [n_A], i\in [k_A]\},\{F_{w,j}^{(3)}:w\in [n_B], j\in [k_B]\}).
\end{align*}

Since $\mathscr{S}_1$ is a local $\epsilon_1$-dilation of $\mathscr{S}_2$, there are isometries $V_A\colon \cl H_{A,2}\to \cl H_{A,1} \otimes \cl K_{A,1}$ and $V_B\colon \cl H_{B,2}\to \cl H_{B,1} \otimes \cl K_{B,1}$ for some finite-dimensional Hilbert spaces $\cl K_{A,1}, \cl K_{B,1}$, and there exist a quantum state $\psi_{\mathrm{junk},1}\in \cl K_{A,1}\otimes \cl K_{B,1}$ such that for all $v,w,i,j$, \begin{align}\label{eq:Vapprox}
(V_A\otimes V_B)(E_{v,i}^{(2)}\otimes F_{w,j}^{(2)})\psi_2 \approx_{\epsilon_1} (E_{v,i}^{(1)}\otimes F_{w,j}^{(1)})\psi_1\otimes \psi_{\mathrm{junk},1}.
\end{align} Similarly, since $\mathscr{S}_2$ is a local $\epsilon_2$-dilation of $\mathscr{S}_3$, there are isometries $W_A\colon \cl H_{A,3}\to \cl H_{A,2} \otimes \cl K_{A,2}$ and $W_B\colon \cl H_{B,3}\to \cl H_{B,2} \otimes \cl K_{B,2}$ for some finite-dimensional Hilbert spaces $\cl K_{A,2}, \cl K_{B,2}$, and there exist a quantum state $\psi_{\mathrm{junk},2}\in \cl K_{A,2}\otimes \cl K_{B,2}$ such that for all $v,w,i,j$, \begin{align}\label{eq:Wapprox}
(W_A\otimes W_B)(E_{v,i}^{(3)}\otimes F_{w,j}^{(3)})\psi_3 \approx_{\epsilon_2} (E_{v,i}^{(2)}\otimes F_{w,j}^{(2)})\psi_2\otimes \psi_{\mathrm{junk},2}.
\end{align}

Define $U_A = (V_A\otimes I_{\cl K_{A,2}})\circ W_A$, and $U_B = (V_B\otimes I_{\cl K_{B,2}})\circ W_B$, and $\psi_{\mathrm{junk}} = \psi_{\mathrm{junk},1}\otimes \psi_{\mathrm{junk},2}$. Clearly $U_A$ and $U_B$ are isometries and that for all $v,w,i,j$ \begin{align*}
(U_A\otimes U_B)(E_{v,i}^{(3)} &\otimes F_{w,j}^{(3)})\psi_3 \\
&= (V_A\otimes I_{\cl K_{A,2}} \otimes V_B \otimes I_{\cl K_{B,2}})(W_A\otimes W_B)(E_{v,i}^{(3)}\otimes F_{w,j}^{(3)})\psi_3 \\
&\approx_{\epsilon_2} (V_A\otimes V_B \otimes I_{\cl K_{A,2}}\otimes I_{\cl K_{B,2}})((E_{v,i}^{(2)}\otimes F_{w,j}^{(2)})\psi_2\otimes \psi_{\mathrm{junk},2}) \\
&\approx_{\epsilon_1} (E_{v,i}^{(1)}\otimes F_{w,j}^{(1)})\psi_1\otimes \psi_{\mathrm{junk},1}\otimes \psi_{\mathrm{junk},2},
\end{align*} where we used \eqref{eq:Wapprox} in the second approximation, and \eqref{eq:Vapprox} in the third approximation. \end{proof}

While working with synchronous quantum correlations one usually takes the quantum state to be full Schmidt rank. We show that this reduction step can also be captured with local dilation. We first need a lemma which is folklore.

\begin{lemma}\label{lem:invariant-subspace}
Let $X\in \bb M_{d_A}, Y\in \bb M_{d_B}$ and $\psi\in \bb C^{d_A}\otimes \bb C^{d_B}$ be such that $(X\otimes I_{d_B})\psi = (I_{d_A}\otimes Y)\psi$. Then, $\mathrm{supp}_A(\psi)$ is an invariant subspace of $X$, and $\mathrm{supp}_B(\psi)$ is an invariant subspace of $Y$.
\end{lemma}

\begin{proof}
If $\psi = \sum_{i=1}^r \alpha_i\xi_i\otimes\eta_i$ is a Schmidt decomposition, then the matrix $D = \sum_{i=1}^r \alpha_i\xi_i\eta_i^T$ satisfies $\mathrm{vec}(D) = \psi$. Note that $\mathrm{supp}_A(\psi)$ is the column space of $D$. To show that $\mathrm{supp}_A(\psi)$ is invariant under $X$, it suffices to show that $X\xi_i \in \mathrm{supp}_A(\psi)$ for all $i\in [r]$. Using Identity \eqref{eq:vec-identity}, we see that $(X\otimes I_{d_B})\psi = (I_{d_A}\otimes Y)\psi$ is equivalent to $XD = DY^T$. Thus, for $i\in [r]$, 
\begin{align*}
  X\xi_i = X\left(\alpha_i^{-1}D\overline{\eta}_i \right) = \alpha_i^{-1}(XD)(\overline{\eta}_i) =  \alpha_i^{-1}(DY^T)(\overline{\eta}_i) \in \mathrm{supp}_A(\psi),
\end{align*} as required. The invariance of $\mathrm{supp}_B(\psi)$ under $Y$ follows similarly.	
\end{proof}

\begin{lemma}\label{lem:full-schmidt-rank}
Let $p=(p(i,j|v,w))\in C_q^s(n,k)$ be a synchronous quantum correlation. Let $\mathscr{S} = (\psi\in \bb C^{d_A}\otimes \bb C^{d_B},\{E_{v,i}\},\{F_{w,j}\})$ be a strategy inducing $p$. Then, there exists a strategy $\mathscr{S}' = (\psi'\in \bb C^r\otimes \bb C^r,\{E_{v,i}'\},\{F_{w,j}'\})$ where $\psi'$ is of full Schmidt rank and $\mathscr{S}'$ is a local dilation of $\mathscr{S}$.
\end{lemma}

\begin{proof}
Consider a Schmidt decomposition $\psi = \sum_{l=1}^r \alpha_l\xi_l\otimes \eta_l$, and let $\iota_A\colon  \bb C^r \to \bb C^{d_A}$ and $\iota_B\colon \bb C^r\to \bb C^{d_B}$ be isometries given by $\iota_A = \sum_{l=1}^r \xi_le_l^*$ and $\iota_B = \sum_{l=1}^r \eta_le_l^*$. Then, the operators \begin{align}
E_{v,i}' = \iota_A^*E_{v,i}\iota_A, \qquad F_{w,j}' = \iota_B^*F_{w,j}\iota_B,
\end{align} are themselves positive, and moreover, for each $v$ and $w$, the sets $\{E_{v,i}':i\in [k]\}$ and $\{F_{w,j}':j\in [k]\}$ form POVMs. With $\psi' = (\iota_A^*\otimes \iota_B^*)\psi$, these above operators constitute the strategy $\mathscr{S}'$, and clearly $\psi'$ is of full Schmidt rank.

To see that $\mathscr{S}'$ is a local dilation of $\mathscr{S}$, set $\cl K_A = \bb C^{d_A}$ and $\cl K_B=\bb C^{d_B}$ and define isometries $V_A\colon \bb C^{d_A} \to \bb C^r \otimes \bb C^{d_A}$ and $V_B\colon \bb C^{d_B} \to \bb C^r \otimes \bb C^{d_B}$, by \begin{align*}
V_A(\xi) &= \iota_A^*(\xi)\otimes \xi_1 + e_1 \otimes (I_{d_A} - \iota_A\iota_A^*)(\xi), \\
V_B(\eta) &= \iota_B^*(\eta)\otimes \eta_1 + e_1 \otimes (I_{d_B} - \iota_B\iota_B^*)(\eta),
\end{align*} where $\xi\in \bb C^{d_A}$ and $\eta\in \bb C^{d_B}$. By Corollary \ref{cor:proj-strat}(a), we have $(E_{v,i}\otimes I_{d_B})\psi = (I_{d_A}\otimes F_{v,i})\psi$ for all $v,i$. Then using Lemma \ref{lem:invariant-subspace}, $\mathrm{supp}_A(\psi)$ is invariant under each $E_{v,i}$, and $\mathrm{supp}_B(\psi)$ is invariant under each $F_{w,j}$, and since $\iota_A\iota_A^*$ is the projection onto $\mathrm{supp}_A(\psi)$ and $\iota_B\iota_B^*$ is the projection onto $\mathrm{supp}_B(\psi)$, we have
\begin{align*}
(V_A\otimes V_B)(E_{v,i}\otimes F_{w,j})\psi = (\iota_A\otimes \iota_B)^*((E_{v,i}\otimes F_{w,j})\psi) \otimes (e_1\otimes e_1).
\end{align*} But then, using the invariance property, \begin{align*}
(\iota_A\otimes \iota_B)^*((E_{v,i}\otimes F_{w,j})\psi) &= (\iota_A^*\otimes \iota_B^*)(\iota_A\iota_A^*\otimes \iota_B\iota_B^*)(E_{v,i}\otimes F_{w,j})\psi \\
&= (\iota_A^*\otimes \iota_B^*)(\iota_A\iota_A^*\otimes \iota_B\iota_B^*)(E_{v,i}\otimes F_{w,j})(\iota_A\iota_A^*\otimes \iota_B\iota_B^*)\psi \\
&= ((E_{v,i}'\otimes F_{w,j}')\psi'),
\end{align*} as desired. \end{proof}

\section{Correlations from projections adding up to scalar times identity}\label{sec:proj-sum-scalars}
We are interested in self-testing projections which sum up to (some specific) scalar times the identity. Such kind of projections have been studied in detail in \cite{KRS}, and we collect some of the material from there as needed.

For $n\in \bb N$, let $\Sigma_n$ be the set of all scalars $x$ such that there exist $n$ projections $R_1,\dots,R_n\in \bh$, for some Hilbert space $\cl H$ (possibly infinite dimensional), such that $R_1+\dots+R_n = x I_{\cl H}$. For $n\leq 3$, it can be shown that \begin{align*}
\Sigma_1 = \{0,1\}, \qquad \Sigma_2 = \{0,1,2\}, \qquad \Sigma_3 = \left\lbrace 0,1,\frac{3}{2},2,3 \right\rbrace.
\end{align*} For $n\geq 4$, the set $\Sigma_n$ is described in \cite{KRS}. In particular, they show that the set $\Sigma_4$ is countably infinite; whereas for all $n\geq 5$, the set $\Sigma_n$ is uncountable with a nondegenerate interval subset. For our purpose, we do not require the full description of the set $\Sigma_n$. The following theorem is a distillation of Proposition~2, Theorem~3, and Theorem~4 from the aforementioned work.

For a definition of a universal C$^*$-algebra we refer the reader to Appendix \ref{appendixA}.

\begin{theorem}\label{thm:KRS}
For $n = 3$, set $\Lambda_3 = \{\frac{3}{2}\}$; and for $n\geq 4$, set $\Lambda_n =\{x_l\}_{l=0}^{\infty}$ with the sequence defined recursively as follows: $x_0 = 0$, and for all $l\geq 1$, \begin{align}\label{eq:recurse-def}
x_l = 1+ \frac{1}{n-1-x_{l-1}}.
\end{align} Then $\Lambda_n \subseteq \Sigma_n$ for all $n\geq 3$. Moreover, if $\mathscr P_{n,x}$ is the universal $C^*$-algebra with generators $r_1,\dots,r_n$ and relations given as follows (here $1$ is the unit of the algebra)\begin{align}\label{eq:relations}
\mathscr P_{n,x} \coloneqq \mathrm{C}^*\left\langle r_1,\dots,r_n \,\big|\,  r_v=r_v^*=r_v^2,\;\forall\; v \in [n], \sum_{v=1}^nr_v = x 1  \right\rangle,
\end{align} then for $x\in \Lambda_n$ ($n\geq 3$), there is a unique irreducible representation (up to unitary equivalence) of the $C^*$-algebra $\mathscr P_{n,x}$. Furthermore, this representation is of dimension $d$, where $x=\frac{b}{d}$ is in its lowest terms.
\end{theorem}

The C$^*$-algebra $\mathscr P_{n,x}$ has the universal property that whenever there exist projections $R_1,\dots,R_n\in \bh$ with $\sum_{v=1}^nR_v= x I_{\cl H}$, then there exists a representation $\pi\colon \mathscr P_{n,x}\to \bh$ with $\pi(r_v)=R_v$ for all $v\in [n]$.

If $x\in \{x_l\}_{l=0}^{\infty}$ with $x = \frac{b}{d}$ in its lowest terms, the uniqueness property of $\mathscr{P}_{n,x}$ implies that if $\widetilde{P}_1,\dots,\widetilde{P}_n$ are projections in $\bb M_d$ satisfying $\widetilde{P}_1+\dots+\widetilde{P}_n = xI_d$, and $R_1,\dots,R_n\in \bh$ are projections such that $R_1+\dots+R_n = x I_{\cl H}$, then there exists a Hilbert space $\cl K$ and a unitary $U \colon \cl H \rightarrow \bb C^d \otimes \cl K$ such that $UR_vU^* = \widetilde{P}_v \otimes I_{\cl K}$ for all $v\in [n]$.

We are now ready to define the synchronous quantum correlations that we will be self-testing.

\begin{definition}\label{def:def-of-p*}
Fix $n\geq 3$ and $x\in \Lambda_n$ with $x = \frac{b}{d}$ in its lowest terms. Let $\widetilde{P}_1,\dots,\widetilde{P}_n$ be projections in $\bb M_{d}$ (by Theorem \ref{thm:KRS}) such that $\widetilde{P}_1+\dots+\widetilde{P}_n=xI_d$. Consider the synchronous correlation $\widetilde{p}_{n,x}=(\widetilde{p}_{n,x}(i,j|v,w))\in C^s_q(n,2)$ induced by the strategy 
\begin{align}\label{eq:can-strat}
\widetilde{\mathscr S} \coloneqq \left(\varphi_d\in \bb C^d\otimes \bb C^d, \left\lbrace\widetilde{P}_v, I_{d}-\widetilde{P}_v\right\rbrace_{v=1}^n, \left\lbrace\widetilde{P}_v^T, I_{d}-\widetilde{P}_v^T \right\rbrace_{v=1}^n\right),
\end{align} where $\varphi_d$ is the maximally entangled state. We shall refer to this strategy as the \emph{canonical strategy} for $\widetilde{p}_{n,x}$.
\end{definition}

We shall use the uniqueness result from Theorem \ref{thm:KRS} to show that the strategy $\widetilde{\mathscr{S}}$ given in Definition \ref{def:def-of-p*} is essentially the unique strategy which induces the correlation $\widetilde{p}_{n,x}$---more precisely, $\widetilde{p}_{n,x}$ self-tests the strategy $\widetilde{\mathscr{S}}$. This result, together with robustness, will be the focus of next section.  First we show how the uniqueness result provides an explicit description of the C$^*$-algebra $\mathscr P_{n, x}$.

\begin{lemma}\label{lem:proj-gen-Md}
Let $n\geq 3$ and $x\in \Lambda_{n}$ with $x=\frac{b}{d}$ in lowest terms. Then $\mathscr{P}_{n,x} \simeq \bb M_d$. In particular, if $R_1,\dots,R_n\in \bb M_d$ are projections such that $\sum_{v=1}^n R_v = x I_d$, then $\{R_v\}_{v=1}^n$ generate $\bb M_d$ as a $C^*$-algebra.
\end{lemma}

\begin{proof}
By Theorem \ref{thm:KRS}, there is an irreducible representation $\pi \colon \mathscr P_{n, x} \to \mathfrak B(\bb C^d)$.  As $\pi$ is irreducible, $\pi$ must be surjective (using that $\mathbb C^d$ is finite-dimensional) \cite[Theorem 3.2.2]{EtingofRepTheory}.  We will show that $\pi$ is injective and hence is an isomorphism.  If this is not the case, then $\ker(\pi)$ is a non-zero C$^*$-algebra and hence admits an irreducible representation $\phi_0 \colon  \ker(\pi) \to \mathfrak B(\mathcal H)$ on a (possibly infinite-dimensional) Hilbert space $\mathcal H$ \cite[Theorem 5.1.12]{Murphy}.  Then $\phi_0$ extends to an irreducible representation $\phi \colon \mathscr P_{n, x} \rightarrow \mathfrak B(\mathcal H)$ \cite[Theorem 5.1.13]{Murphy} which is necessarily not unitarily equivalent to $\pi$ since $\ker(\pi) \neq \ker(\phi)$.  This contradicts the uniqueness of the irreducible representation of $\mathscr P_{n, x}$ and hence shows that $\pi$ is injective.
\end{proof}

We now collect some further consequences of the uniqueness property in Theorem \ref{thm:KRS}. The following lemma tells us that the projections $\{R_v\}_{v=1}^n$ obtained through a representation of $\mathscr P_{n,x}$ have the same trace and are ``symmetrically distributed''.

\begin{lemma}\label{prop:const-trace}
	Let $n\geq 3$ and $x\in \Lambda_{n}$ with $x=\frac{b}{d}$ in its lowest terms. Let $R_1,\dots,R_n \in \bb M_d$ be projections such that $\sum_{v=1}^n R_v = x I_d$. Then, \begin{align}\label{eq:proj-trace}
	\mathrm{tr}_{d}(R_vR_w) = \begin{cases}
	\frac{x}{n} &\text{ if } v=w, \\
	\frac{x(x-1)}{n(n-1)} &\text{ if } v\neq w.
	\end{cases}
\end{align} More generally, if $R_1,\dots,R_n\in \bb M_k$ (for some $k\in \bb N$) are projections such that $\sum_{v=1}^n R_v = x I_k$, then $\tr(R_v)$ is constant for all $v\in [n]$, and $\tr(R_vR_w)$ is constant for all $v\neq w$.
\end{lemma}

\begin{proof}
	Since $\sum_{v=1}^n R_v = x I_d$, by the universal property of $\mathscr{P}_{n,x}$, there exists a representation $\pi\colon\mathscr{P}_{n,x} \to \bb M_d$ such that $\pi(r_v) = R_v$ for all $v\in [n]$. By Theorem \ref{thm:KRS} this representation is unique (up to unitary equivalence) and irreducible.
	
	Now, choose a pair $(v,w)$ with $v\neq w$ and define a new representation $\sigma_{v,w}\colon \mathscr P_{n,x}\to \bb M_{d}$ by \begin{align*}
	\sigma_{v,w}(r_k) &= \begin{cases}
	R_k &\text{ if } k\neq v \text{ and } k\neq w, \\
	R_w &\text{ if } k = v, \\
	R_v &\text{ if } k = w. \\
	\end{cases}
	\end{align*} By the uniqueness of the irreducible representation $\pi$ (Theorem \ref{thm:KRS}), there exists a unitary $U\in \bb M_{d}$ such that $\sigma_{v,w}(r_k) = U\pi(r_k) U^*$ for all $k$. But then, \begin{align*}
	\tr(R_v) = \tr(\sigma_{v,w}(r_w)) = \tr(U\pi(r_w)U^*) = \tr(\pi(r_w)) = \tr(R_w).
	\end{align*} Additionally, since $R_1+\dots+R_n = xI_d$, it follows that $\mathrm{tr}_{d}(R_v)=\frac{x}{n}$ for all $v\in [n]$.
	
	Similarly, one can show that $\tr(R_vR_w)$ is constant for all $v\neq w$. Moreover, squaring $\sum_{v=1}^n R_v = xI_d$, and taking normalized trace yields that $\mathrm{tr}_{d}(R_vR_w) = \frac{x(x-1)}{n(n-1)}$, for all $v\neq w$.
	
	For the final statement, let $\pi \colon \mathscr P_{n, x} \rightarrow \bb M_k$ be the representation given by $\pi(r_v) = R_v$ for all $v \in [n]$.  If $\pi_0 \colon \mathscr P_{n, x} \rightarrow \bb M_d$ is the unique irreducible representation, then $d$ divides $k$ and $\pi$ is unitarily equivalent to $\pi_0^{\oplus (k/d)}$.  So, $\mathrm{Tr}(R_vR_w) = k \mathrm{tr}_d(\pi_0(r_v r_w))$ for all $v, w \in [n]$. The result follows from the first part of the lemma.
\end{proof}

\begin{remark}[The correlation table for $\widetilde{p}_{n,x}$]
Let $n\geq 3$ and $x\in \Lambda_n$ with $x=\frac{b}{d}$ in lowest terms. Let $\widetilde{p}_{n,x}$ be the quantum correlation as in Definition \ref{def:def-of-p*}. Then, \begin{align}\label{eq:pstardef}
	\widetilde{p}_{n,x}(1,1|v,w) = \inner{\left(\widetilde{P}_v\otimes \widetilde{P}_w^T\right)\varphi_d}{\varphi_d} =  \mathrm{tr}_{d}(\widetilde{P}_v\widetilde{P}_w)  = \begin{cases}
	\frac{x}{n} &\text{ if } v=w, \\
	\frac{x(x-1)}{n(n-1)} &\text{ if } v\neq w,
	\end{cases}
	\end{align} where the second equality follows from Identity~\eqref{eq:vec-iden-32} and the third equality follows from Lemma~\ref{prop:const-trace}. Note that with these values of $\widetilde{p}_{n,x}(1,1|v,w)$, all other values of $\widetilde{p}_{n,x}(i,j|v,w)$ can be deduced using the synchronous condition. Indeed, for all $v \neq w$, \begin{align*}
	\widetilde{p}_{n,x}(1,2|v,w) = \widetilde{p}_{n,x}(2,1|v,w) &= \frac{x}{n} - \widetilde{p}_{n,x}(1,1|v,w) \text{ and} \\
	\widetilde{p}_{n,x}(2,2|v,w) &= 1 - \frac{2x}{n} + \widetilde{p}_{n,x}(1,1|v,w),
	\end{align*} and for $v = w$, \begin{align*}
   \widetilde{p}_{n,x}(1,2|v,v) &= \widetilde{p}_{n,x}(2,1|v,v) = 0 \text{ and} \\
	\widetilde{p}_{n,x}(2,2|v,v) &= 1 - \widetilde{p}_{n,x}(1,1|v,v).
\end{align*}
\end{remark}

For $n\geq 3$ and $x=\frac{n}{n-1}\in \Lambda_n$, Lemma \ref{prop:const-trace} implies that $\tr(R_v)=1$ for projections $R_1,\dots,R_n\in \bb M_{n-1}$ satisfying $\sum_{v=1}^n R_v = \frac{n}{n-1}I_{n-1}$. Thus, each $R_v$ is the projection onto the span of some unit vector $\xi_v\in \mathbb{C}^{n-1}$. The following corollary is then straightforward to verify.

\begin{cor}\label{cor:simplex}
Let $n\geq 3$. Projections $\{R_v\}_{v=1}^n\subseteq \bb M_{n-1}$ which sum up to $\frac{n}{n-1}I_{n-1}$ give rise to $n$ unit vectors $\{\xi_v\}_{v=1}^n\subseteq \bb C^{n-1}$ such that $|\inner{\xi_v}{\xi_w}|=\frac{1}{n-1}$ for all $v\neq w$, and conversely.
\end{cor}

Together with unitary equivalence, the property that $|\inner{\xi_v}{\xi_w}|=\frac{1}{n-1}$ for all $v\neq w$ implies that the vectors can be taken as the vertices of a regular $(n-1)$-simplex in $\mathbb{R}^{n-1}$ centered at the origin. For instance, in dimension $2$, (i.e., $n=3$ in Corollary \ref{cor:simplex}), the vertices of a regular $2$-simplex (a triangle) centered at origin are \begin{align*}
\xi_1 = \begin{bmatrix}
1 \\ 0
\end{bmatrix}, \qquad \xi_2 = \frac{1}{2}\begin{bmatrix}
1 \\ -\sqrt{3}
\end{bmatrix}, \qquad \xi_3 = \frac{1}{2}\begin{bmatrix}
-1 \\ -\sqrt{3}
\end{bmatrix}.
\end{align*} Then the projections $R_1, R_2, R_3$ are onto the span of $\xi_1,\xi_2,\xi_3$, respectively, and are given by \begin{align}
	R_1 = \begin{bmatrix}
	1 & 0 \\ 0 & 0
	\end{bmatrix}, \quad R_2 = \frac{1}{4} \begin{bmatrix}
	1 & -\sqrt{3} \\ -\sqrt{3} & 3
	\end{bmatrix}, \quad R_3 = \frac{1}{4} \begin{bmatrix}
	1 & \sqrt{3} \\ \sqrt{3} & 3
	\end{bmatrix}.
\end{align}

For other $x\in \Lambda_n$ ($n\geq 4$), the projections adding up to $xI$ can be constructed using a recursive method described in \cite{KRS}, but it is not clear what natural geometric picture we may associate to such projections. For the sake of completeness, we describe this construction for $n=4$ in Appendix~\ref{appendixB}.

Finally, we require a result on the maximum eigenvalue of a certain matrix obtained from projections adding up to a scalar times identity. 

\begin{lemma}\label{lem:max-evalue}
Let $n\geq 3$ and $x\in \Lambda_{n}$ with $x=\frac{b}{d}$ in lowest terms. Let $R_1,\dots,R_n \in \bb M_d$ be projections such that $\sum_{v=1}^n R_v = x I_d$. If $N = \sum_{v=1}^n R_v\otimes R_v^T$, then $x$ is the largest eigenvalue of $N$, and further, the eigenspace corresponding to $x$ is the one-dimensional space spanned by the maximally entangled state $\varphi_d$.
\end{lemma}

\begin{proof}
	With $D = \frac{1}{\sqrt{d}}I_d$, we have $\mathrm{vec}(D) = \varphi_d$. Using Identity \eqref{eq:vec-identity}, \begin{align*}
		N\varphi_d = N\mathrm{vec}(D) = \sum_{v=1}^n \mathrm{vec}(R_vDR_v) = x\mathrm{vec}(D) = x\varphi_d.
	\end{align*} Thus $\varphi_d$ is an eigenvector of $N$ with eigenvalue $x$.
	
	Since $N\geq 0$, the largest eigenvalue of $N$ is $\|N\|$. Since \begin{align*}
		N = \sum_{v=1}^n R_v\otimes R_v^T \leq \sum_{v=1}^n R_v \otimes I_d = x I_d\otimes I_d
	\end{align*} it follows that $\|N\|\leq x$. Then $\|N\| = \sup_{\|\xi\| = 1}|\inner{N\xi}{\xi}|$ coupled with $\inner{N\varphi_d}{\varphi_d} = x$, yields $\|N\| = x$.
	
	We now show that the eigenspace corresponding to the eigenvalue $x$ is one-dimensional. Let $\varphi\in \bb C^d\otimes \bb C^d$ be a unit vector such that $N\varphi = x\varphi$. Let $B\in \bb M_d$ be such that $\mathrm{vec}(B) = \varphi$. Then $N\varphi = x\varphi$ is equivalent to $N\mathrm{vec}(B) = x\mathrm{vec}(B)$, which in turn, using Identity \eqref{eq:vec-identity}, is equivalent to \begin{align}\label{eq:fixed-point}
		\sum_{v=1}^n R_vBR_v = xB.
	\end{align} Now consider the unital quantum channel $\Phi\colon \bb M_d\to \bb M_d$ defined by \begin{align}
		X\mapsto \frac{1}{x}\sum_{v=1}^n R_vXR_v, \qquad X\in \bb M_d.
	\end{align} Then Equation \eqref{eq:fixed-point} implies that $B$ is a fixed point of the quantum channel. By \cite[Theorem 4.25]{Watr18}, it follows that $R_vB= BR_v$ for all $v\in [n]$. By Lemma \ref{lem:proj-gen-Md}, the projections $\{R_v\}_{v=1}^n$ generate the whole of $\bb M_d$. Hence, $B$ commutes with every element of $\bb M_d$, and we must have $B = \lambda I_d$, for some $\lambda\in\bb C$. In other words, $\varphi = \mathrm{vec}(B) = (\lambda\sqrt{d})\varphi_d$, as required. \end{proof}

\section{Robust self-testing of projections adding up to scalar times identity}\label{self-testing-sum-proj}
We now set out to prove that the correlation $\widetilde{p}_{n,x}$ of Definition \ref{def:def-of-p*} robustly self-tests the canonical strategy $\widetilde{\mathscr S}$ therein.

We first show in Theorem~\ref{thm:self-test} that if a strategy $\mathscr{S}$ induces the correlation $\widetilde{p}_{n,x}$ then $\widetilde{\mathscr S}$ must be a local dilation of $\mathscr{S}$. For this, we prove that the measurement operators of $\mathscr{S}$ form a representation of the algebra $\mathscr{P}_{n,x}$. Then invoking the irreducibility assumption we get isometries relating the strategy $\mathscr S$ in hand with the canonical one. 

To handle the robust case, we first formally define a suitable notion of an ``approximate'' representation. Then we show that if we have a strategy which induces a correlation within $\epsilon$-distance of $\widetilde{p}_{n,x}$, then we get an ``approximate'' representation of the relation $r_1 + \dots + r_n = x1$. Finally, to relate the approximate strategy with the canonical one, we establish an analogue of Gowers--Hatami Theorem for the C$^*$-algebra $\mathscr{P}_{n,x}$.

\begin{theorem}\label{thm:self-test}
Let $n\geq 3$ and $x\in \Lambda_n$. The synchronous correlation $\widetilde{p}_{n,x} \in C_q^s(n,2)$ induced by the strategy $\widetilde{\mathscr S}$ in Definition \ref{def:def-of-p*} self-tests $\widetilde{\mathscr S}$.
\end{theorem}

\begin{proof}
Let $\mathscr{S} = (\psi\in \bb C^{d_A}\otimes \bb C^{d_B},\{P_{v}, I_{d_A}- P_v\}_{v=1}^n,\{Q_{w}, I_{d_B}- Q_w\}_{w=1}^n)$ be a strategy which induces correlation $\widetilde{p}_{n,x}$. Using the transitivity of local dilations (Lemma \ref{lem:trans-epsilon-dilation}) and Lemma \ref{lem:full-schmidt-rank} we may assume that $\psi$ is of full Schmidt rank $r$, i.e., $d_A = d_B = r$. The ``moreover'' part of Corollary \ref{cor:proj-strat} then implies that $P_v$ and $Q_v$ are projections for all $v\in [n]$.
	
Set $\rho = \psi\psi^*$ and let $\rho_A$ and $\rho_B$ be its reduced density matrices. Clearly, $\inner{I_{r}}{I_{r}}_{\rho_A} = 1$, and if $P \coloneqq\sum_{v=1}^n P_v$, then using the property $(P_v\otimes I_{r})\psi = (I_{r}\otimes Q_v)\psi$ of Corollary \ref{cor:proj-strat}, we compute \begin{align*}
	\inner{P}{I_{r}}_{\rho_A} &= \sum_{v=1}^n \inner{(P_v\otimes I_{r})\psi}{\psi} = \sum_{v=1}^n \inner{(P_v\otimes Q_v)\psi}{\psi} =  \sum_{v=1}^n \widetilde{p}_{n,x}(1,1|v,v) = x, \\
\intertext{and}
	\inner{P}{P}_{\rho_A} &= \sum_{v,w=1}^n \inner{(P_vP_w\otimes I_{r})\psi}{\psi}  = \sum_{v,w=1}^n \inner{(P_v\otimes I_{r})\psi}{(P_w\otimes I_{r})\psi} \\
	&= \sum_{v,w=1}^n \inner{(P_v\otimes I_{r})\psi}{(I_{r} \otimes Q_w)\psi} = \sum_{v,w=1}^n \inner{(P_v\otimes Q_w)\psi}{\psi} \\
	&= \sum_{v,w=1}^n \widetilde{p}_{n,x}(1,1|v,w) = n\frac{x}{n} + n(n-1)\frac{x(x-1)}{n(n-1)} = x^2.
	\end{align*} 
Therefore, $\inner{P}{I_{r}}_{\rho_A} = \|P\|_{\rho_A} \|I_r\|_{\rho_A} = x$. That is, equality holds in a Cauchy--Schwarz inequality, and since $\psi$ is of full Schmidt rank, this equality implies $P = \sum_{v=1}^n P_v =  xI_{r}$. Similarly, we can show that $\sum_{w=1}^n Q_w  = x I_{r}$.
	
Invoking the uniqueness of the irreducible representation of $\mathscr{P}_{n,x}$, it follows that there exist unitaries $U_A\colon \bb C^r \to \bb C^d\otimes \cl K_A$ and $U_B:\bb C^r \to \bb C^d \otimes \cl K_B$ for some Hilbert spaces $\cl K_A, \cl K_B$ such that $U_AP_vU_A^* = \widetilde{P}_v\otimes I_{\cl K_A}$ and $U_BQ_vU_B^* = \widetilde{P}_v^T\otimes I_{\cl K_B}$ for all $v\in [n]$. Set $U= U_A\otimes U_B$. Then, we have \begin{align*}
	x &= \sum_{v=1}^n \inner{(P_v\otimes Q_v)\psi}{\psi} \\
	&= \sum_{v=1}^n \inner{(\widetilde{P}_v\otimes \widetilde{P}_v^T\otimes I_{\cl K_A} \otimes I_{\cl K_B})U\psi}{U\psi} \\
	& = \inner{(N \otimes I_{\cl K_A} \otimes I_{\cl K_B})U\psi}{U\psi},
	\end{align*} where $N= \sum_{v=1}^n\widetilde{P}_v\otimes \widetilde{P}_v^T$. Then using Lemma \ref{lem:max-evalue}, $x$ is also the largest eigenvalue of $N\otimes (I_{\cl K_A}\otimes I_{\cl K_B})$ with eigenspace $\{\varphi_d\otimes \psi':\psi'\in \cl K_A\otimes \cl K_B\}$. The above computation implies, by means of Cauchy--Schwarz inequality, that $U\psi$ is an eigenvector of $N\otimes (I_{\cl K_A}\otimes I_{\cl K_B})$ corresponding to the eigenvalue $x$, and therefore $U\psi = \varphi_d\otimes \psi_{\mathrm{junk}}$ for some unit vector $\psi_{\mathrm{junk}}\in \cl K_A\otimes \cl K_B$. Finally, for all $v,w$, \begin{align*}
	U(P_v\otimes Q_w)\psi &= U(P_v\otimes Q_w)U^*U\psi \\
	&= \big((\widetilde{P}_v\otimes \widetilde{P}_w^T)\otimes (I_{\cl K_A}\otimes I_{\cl K_B})\big)(\varphi_d\otimes \psi_{\mathrm{junk}}) \\
	&= \big((\widetilde{P}_v\otimes \widetilde{P}_w^T)\varphi_d\big) \otimes \psi_{\mathrm{junk}},
	\end{align*} as required.
\end{proof}

With an eye towards the robust self-testing result, we first define what we mean by an ``approximate'' representation for a relation. This is an analogue of an ``approximate'' group representation.

\begin{definition}\label{def:approx-rep}
A \emph{polynomial relation} in $n$ variables is a relation of the form \begin{align}\label{polynomial}
	f\left(x_1,\dots,x_n,x_1^*,\dots,x_n^*\right) = 0
\end{align} where $f$ is a polynomial in $2n$ non-commuting variables $x_1,\dots,x_n,x_1^*,\dots,x_n^*$. Given $\delta \geq 0$, a density matrix $\sigma\in \bb M_r$, and polynomial relations $(f_1, \ldots, f_k)$ as in Expression \eqref{polynomial}, an $n$-tuple of matrices $\left(X_1,\dots,X_n\right)$ in $\bb M_r$ is called a $(\delta,\sigma)$-\emph{representation} of relations $(f_1, \ldots, f_k)$ if
\begin{align*}
  \max_{i \in [k]} \norm{f_i\left(X_1,\dots,X_n,X_1^*,\dots,X_n^*\right)}_{\sigma} \leq \delta. 
\end{align*}
\end{definition}

\begin{lemma}\label{thm:almost-corr-nonsyn} Let $n\geq 3$ and $x\in \Lambda_n$. Let $\widetilde{p}_{n,x}\in C_q^s(n,2)$ be as in Definition \ref{def:def-of-p*}. Let $p\in C_q(n,2)$ and suppose it is induced by a strategy 
\begin{align*}
\left(\psi\in \bb C^{d_A}\otimes \bb C^{d_B}, \left\lbrace E_v, I_{d_A}-E_v \right\rbrace_{v=1}^n, \left\lbrace F_w, I_{d_B} - F_w \right\rbrace_{w=1}^n \right).
\end{align*} 
Let $\rho_A$ and $\rho_B$ be the reduced density matrices of $\rho=\psi\psi^*$. If $\|p-\widetilde{p}_{n,x}\|_1 \leq \delta$ for some $\delta \geq 0$, then the tuple $(E_1,\dots,E_n)$ is a $(C\delta^{\frac{1}{4}},\rho_A)$-representation and $(F_1,\dots,F_n)$ is a $(C\delta^{\frac{1}{4}},\rho_B)$-representation of the relations given in \eqref{eq:relations}, where \begin{align}
C \coloneqq \sqrt{n^2 + (1+2x)\sqrt{\delta}}.
\end{align}  \end{lemma}

\begin{proof}
Using triangle inequality in comparing marginals, \begin{align*}
\bigg| \sum_{v=1}^n \big(p_A(1|v) &- (\widetilde{p}_{n,x})_A(1|v)\big) \bigg| \\
&= \left|\sum_{v=1}^n \big(p(1,1|v,v) + p(1,2|v,v) - \widetilde{p}_{n,x}(1,1|v,v) - \widetilde{p}_{n,x}(1,2|v,v)\big) \right| \\
&\leq \sum_{v=1}^n |p(1,1|v,v) - \widetilde{p}_{n,x}(1,1|v,v)| + |p(1,2|v,v) - \widetilde{p}_{n,x}(1,2|v,v)| \\
&\leq \|p - \widetilde{p}_{n,x}\|_1 \leq \delta.
\end{align*} Since $\sum_{v=1}^n (\widetilde{p}_{n,x})_A(1|v) = x$, we get \begin{align}\label{eq:equation}
\sum_{v=1}^n p_A(1|v) \geq x - \delta.
\end{align}

Similarly, we establish \begin{align}\label{ineq:fromthm5.4}
\left|x^2 - \sum_{v,w=1}^n p(1,1|v,w) \right| \leq \delta.
\end{align}

On the other hand, using triangle inequality and Lemma \ref{lem:approx-vectors}(a), \begin{align*}
\sum_{v,w=1}^n \inner{(E_vE_w\otimes I_{d_B})\psi}{\psi} &= \sum_{v,w=1}^n \inner{(E_v\otimes I_{d_B})\psi}{(E_w\otimes I_{d_B})\psi} \\
&\approx_{n^2\sqrt{\delta}} \sum_{v,w=1}^n \inner{(E_v\otimes I_{d_B})\psi}{(I_{d_A}\otimes F_w)\psi} \\
&= \sum_{v,w=1}^n \inner{(E_v\otimes F_w)\psi}{\psi} = \sum_{v,w=1}^n p(1,1|v,w).
\end{align*} Therefore, \begin{align}
\left|\sum_{v,w=1}^n \inner{(E_vE_w\otimes I_{d_B})\psi}{\psi} - \sum_{v,w=1}^n p(1,1|v,w) \right| \leq n^2 \sqrt{\delta}.
\end{align} Using this inequality together with Inequality \eqref{ineq:fromthm5.4} and an application of triangle inequality, we get \begin{align}\label{ineq:one-more}
\left|x^2 - \sum_{v,w=1}^n \inner{(E_vE_w\otimes I_{d_B})\psi}{\psi} \right|  \leq \delta + n^2\sqrt{\delta}.
\end{align}

Finally, if $E \coloneqq E_1 + \dots + E_n$, then using Inequalities \eqref{eq:equation} and \eqref{ineq:one-more}, \begin{align*}
\|E - x I_{d_A}\|_{\rho_A}^2 &= \sum_{v,w=1}^n \inner{(E_vE_w\otimes I_{d_B})\psi}{\psi} - 2x \sum_{v=1}^n p_A(1|v) + x^2 \\
&\leq x^2 + \delta + n^2\sqrt{\delta} - 2x(x - \delta) + x^2 = (1+2x)\delta + n^2\sqrt{\delta}.
\end{align*} Similarly, one proves the statement about $F = F_1+\dots+F_n$. \end{proof}

We now supply with an analogue of Gowers--Hatami Theorem as promised. Our presentation of Gowers--Hatami Theorem mimics the one presented in \cite[Theorem 12]{Vidick} (which is a slightly general version of \cite[Theorem 15.2]{Gowers17} and was originally published in \cite{GH17}). We remark that the relation between $\epsilon$ and $\delta$ in Theorem \ref{thm:Gowers-Hatami-N} is non-constructive. We first recall a standard fact from C$^*$-algebra theory.

\begin{lemma}
If $\tau$ is a tracial state on a unital $C^*$-algebra $\mathfrak A$, then $\mathfrak N_\tau = \{a \in \mathfrak A : \tau(a^*a)=0\}$ is a closed two-sided ideal of $\mathfrak A$.
\end{lemma}

\begin{proof}
Applying the Cauchy--Schwarz inequality to the sesquilinear form $(a, b) \mapsto \tau(b^*a)$ on $\mathfrak A$ shows $\mathfrak N_\tau$ is a subspace of $\mathfrak A$, which is closed as $\tau$ is continuous. For $z \in \mathfrak A$ and $a \in \mathfrak N_\tau$, using the inequality $z^*z\leq \|z^*z\|1_{\ff A}$, we have \begin{align*}
\tau((za)^*(za)) = \tau(a^*z^*za) \leq \|z^*z\|\tau(a^*a) = 0,
\end{align*} and hence $za \in \mathfrak N_\tau$.  This shows $\mathfrak N_\tau$ is a left-ideal. As $\tau(a^*a) = \tau(aa^*)$ for all $a \in \mathfrak A$, $\mathfrak N_\tau$ is self-adjoint and hence is also a right ideal.
\end{proof}

\begin{theorem}[Analogue of Gowers--Hatami]\label{thm:Gowers-Hatami-N}
Let $n \geq 3$ and $x \in \Lambda_n$ with $x = \frac{b}{d}$ in lowest terms. Let $\{\widetilde{P}_v\}_{v=1}^n\subseteq \bb M_d$ be an irreducible representation of $\mathscr P_{n,x}$. Then, given $\epsilon > 0$, there exist $\delta > 0$ and $m \in \bb N$ such that the following property is satisfied:
\begin{equation}\label{GH-condition}\begin{minipage}{.9\textwidth}
for all $r\in \bb N$, density matrix $\rho \in \mathbb M_r$, and positive contractions $E_1, \dots, E_n \in \bb M_r$ satisfying
\begin{enumerate}
  \item[(a)] $\norm{E_v^2 - E_v}_\rho \leq \delta$ for all $v \in  n$,
  \item[(b)] $\norm{xI_r - \sum_{v=1}^n  E_v}_\rho \leq \delta$, and
  \item[(c)] $\abs{\tr_{\rho}(W_1W_2 - W_2W_1)} \leq \delta$ for monomials $W_1$ and $W_2$ of degree at most $m$ in the noncommuting variables $E_1, \dots, E_n$,
\end{enumerate} 
there are $s\in \bb N$ and an isometry $V \colon \mathbb C^r \to \mathbb C^d \otimes \mathbb C^s$ such that for all $v \in [n]$, we have $\big\|E_v - V^*(\widetilde P_v \otimes I_s)V\big\|_\rho \leq \epsilon.$
\end{minipage}
\end{equation}
\end{theorem}

\begin{proof}
Suppose the result is false and fix a counterexample $n \geq 3$,  $x \in \Lambda_n$, and $\epsilon > 0$. For each natural number $k \geq 1$, fix $\delta_k > 0$ and $m_k\in \bb N$ such that $\lim_{k\to \infty} \delta_k = 0$ and $\lim_{k\to \infty} m_k = \infty$. For each $k \geq 1$, fix $r_k \in \mathbb N$, a density matrix $\rho_k \in \mathbb M_{r_k}$, and positive contractions $E_{1, k},\dots,E_{n,k} \in \bb M_{r_k}$, such that
\begin{enumerate}
  \item[(a$_k$)] $\big\|E_{i,k}^2 - E_{i,k}\big\|_{\rho_k} \leq \delta_k$ for all $i \in [n]$,
  \item[(b$_k$)] $\norm{x I_{r_k} - \sum_{i=1}^n  E_{i, k}}_{\rho_k} \leq \delta_k$, and
  \item[(c$_k$)] $\abs{\tr_{\rho_k}(W_1W_2-W_2W_1)} \leq \delta_k$ for monomials $W_1,W_2$ of degree at most $m_k$ in the variables $E_{1,k}, \dots, E_{n,k}$,
\end{enumerate} but such that for all natural numbers $s_k \geq 1$ and isometries $V_k \colon \mathbb C^{r_k} \to \mathbb C^d \otimes \mathbb C^{s_k}$, there is an $i \in [n]$ with  \begin{align}\label{eq:contr-eq}
\big\|E_{i,k} - V_k^*(\widetilde P_{i} \otimes I_{s_k})V_k\big\|_{\rho_k} > \epsilon.
\end{align}

Let $\mathfrak A$ be the universal unital C$^*$-algebra generated by $n$ contractions $A_1, \dots, A_n$. For each $k\in \bb N$, by the universal property of $\mathfrak A$ there exists a representation $\pi_k\colon \mathfrak A \to \bb M_{r_k}$ such that $\pi_k(A_i) = E_{i, k}$ for all $i\in [n]$.  Note that $(\tr_{\rho_k}\circ \pi_k)(AB-BA) \to 0$ for all $A, B \in \mathfrak A$.  If $\tau$ is a weak$^*$-limit point of the set of states $\{\tr_{\rho_k}\circ\pi_k\}_{k=1}^{\infty}$ on $\mathfrak A$, then $\tau$ is a tracial state on $\mathfrak A$.  Passing to a subsequence, we may assume $(\tr_{\rho_k}\circ\pi_k)(A) \to \tau(A)$ for all $A \in \mathfrak A$.

Let $q \colon \mathfrak A \to \mathfrak A / \mathfrak N_\tau$ be the quotient map where $\mathfrak N_\tau = \{a \in \mathfrak A : \tau(a^*a)=0\}$. For $1 \leq i \leq n$, we have \begin{align*}
\tau\left((A_i^2 - A_i)^*(A_i^2 - A_i)\right) &= \lim_{k\to \infty} \left(\tr_{\rho_k}\circ \pi_k \right)\left((A_i^2 - A_i)^*(A_i^2 - A_i)\right) \\
&= \lim_{k\to \infty} \tr_{\rho_k}\left((E_{i,k}^2 - E_{i,k})^*(E_{i,k}^2 - E_{i,k})\right) \\
&= \lim_{k\to\infty} \|E_{i,k}^2 - E_{i,k}\|_{\rho_k}^2 = 0.
\end{align*} So $A_i^2 - A_i \in \mathfrak N_\tau$, and hence $q(A_i)$ is a projection for all $1 \leq i \leq n$.  Similarly, \begin{align*}
\tau\left(\left( x I_{\ff A} - \sum_{i=1}^n A_i\right)^*\left(x I_{\ff A} - \sum_{i=1}^n A_i\right)\right) = 0,
\end{align*} and hence $\sum_{i=1}^n q(A_i) = x I_{\mathfrak A / \mathfrak N_\tau}$.  Therefore, there is a representation $\theta \colon \mathscr P_{n, x} \to \mathfrak A / \mathfrak N_\tau$ such that $\theta(r_i) = q(A_i)$ for $1 \leq i \leq n$ where $r_1,\dots, r_n$ are the generators of $\mathscr P_{n, x}$.

Since $\theta \colon \mathscr P_{n, x} \to \mathfrak A / \mathfrak N_\tau$ is a $*$-homomorphism, there is a unital, completely positive map $\widetilde \theta \colon \mathscr P_{n, x} \to \mathfrak A$ such that $q \widetilde{\theta} = \theta$.\footnote{This is a special case of the Choi--Effros lifting theorem \cite{Choi-Effros}. The special case needed here is \cite[Lemma 3.3]{Arveson} using that $\mathscr P_{n, x}$ is a full matrix algebra by Theorem \ref{lem:proj-gen-Md}.} Define $\widetilde \pi_k\colon \mathscr P_{n, x} \to \bb M_{r_k}$ by $\widetilde \pi_k = \pi_k \circ \widetilde \theta$.  Then $\widetilde \pi_k$ is unital and completely positive.  For each $1 \leq i \leq n$, we have $q(\widetilde\theta(r_i)) = \theta(r_i) =  q(A_i)$, and hence $\widetilde \theta(r_i) - A_i \in \mathfrak N_\tau$.  In particular, \begin{align*}
\lim_{k \to \infty} \norm{\widetilde \pi_k(r_i) - E_{i, k}}_{\rho_k}^2 = \tau\left(\left(\widetilde \theta(r_i) - A_i\right)^*\left(\widetilde \theta(r_i) - A_i\right)\right) = 0.
\end{align*}

For each $k \geq 1$, as $\widetilde \pi_k$ is a unital, completely positive map between finite-dimensional C$^*$-algebras, Stinespring's Theorem \cite[Theorem 4.1]{CBMOA} produces a finite-dimensional Hilbert space $\mathcal H_k$, a $*$-homomorphism $\sigma_k\colon \mathscr P_{n, x} \to \mathfrak B(\mathcal H_k)$, and an isometry $W_k \colon \mathbb C^{r_k} \to \mathcal H_k$ such that $W_k^* \sigma_k(a) W_k = \widetilde{\pi}_k(a)$ for all $a\in \mathscr{P}_{n,x}$. In particular, for each $k \geq 1$ and for all $i \in [n]$, we have \begin{align}\label{eq:GHeq1}
 \lim_{k \to \infty} \norm{W_k^*\sigma_k(r_i) W_k - E_{i,k}}_{\rho_k}  = 0.
\end{align}

Since $\mathscr P_{n, x}$ has a unique irreducible representation $\mathscr P_{n, x} \to \mathbb M_d$ given by $r_i \mapsto \widetilde P_i$,
for each $k \geq 1$, there are a natural number $s_k \geq 1$ and unitary $U_k \colon \mathcal H_k \to \mathbb C^d \otimes \mathbb C^{s_k}$ such that $U_k^*(\widetilde P_i \otimes I_{s_k}) U_k = \sigma_k(r_i)$ for all $i\in [n]$.
Define $V_k = U_k W_k \colon \mathbb C^{r_k} \to \mathbb C^d \otimes \mathbb C^{s_k}$.  Then $V_k$ is an isometry, and for all $i \in [n]$, we can rewrite Equation \eqref{eq:GHeq1} as \begin{align*}
\lim_{k \to \infty} \big\|V_k^* (\widetilde P_i \otimes I_{s_k})V_k - E_{i,k}\big\|_{\rho_k} = 0.
\end{align*} In particular, for some large $k \geq 1$, we have $\big\|V_k^* (\widetilde P_i \otimes I_{s_k}) V_k - E_{i, k}\big\|_{\rho_k} \leq \epsilon$ for all $i \in [n]$, which is a contradiction to Inequality \eqref{eq:contr-eq}.
\end{proof}

\begin{remark}
The proof of Theorem \ref{thm:Gowers-Hatami-N} is by contradiction and therefore the relation between $\epsilon$ and $\delta$ is not known. To find the explicit dependence of $\delta$ on $\epsilon$ is left as an open problem.
\end{remark}

With the analogue of the Gowers--Hatami Theorem in hand, we now proceed to prove the claimed robustness result. We first need three intermediate lemmas.

\begin{lemma}\label{lem:approx-evector}
	Let $A\in \bb M_d$ be a positive matrix having at least two distinct eigenvalues and let $\lambda_1>\lambda_2>\dots>\lambda_l \geq 0$ be the distinct eigenvalues of $A$. Let $\xi\in \bb C^d$ be a unit vector such that $\lambda_1 - \epsilon \leq \inner{A\xi}{\xi}$ for some $\epsilon \geq 0$. If $Q_1$ is the projection onto the eigenspace corresponding to $\lambda_1$, then \begin{align*}
	\norm{Q_1\xi}^2 \geq 1 - \frac{\epsilon}{\lambda_1-\lambda_2}.
	\end{align*}
\end{lemma}

\begin{proof}
By the Spectral Theorem, $A = \sum_{i=1}^l \lambda_i Q_i$ where $Q_i$ be the projection onto the eigenspace of $A$ corresponding to the eigenvalue $\lambda_i$. Then, $\inner{A\xi}{\xi} = \sum_{i=1}^l \lambda_i \norm{Q_i\xi}^2$, and hence \begin{align*}
\epsilon &\geq \lambda_1 - \sum_{i=1}^l \lambda_i \norm{Q_i\xi}^2  = \lambda_1 \sum_{i=1}^l \norm{Q_i\xi}^2 - \sum_{i=1}^l \lambda_i \norm{Q_i\xi}^2 = \sum_{i=2}^l (\lambda_1 - \lambda_i) \norm{Q_i\xi}^2 \\
&\geq (\lambda_1 - \lambda_2) \sum_{i=2}^l \norm{Q_i\xi}^2 = (\lambda_1-\lambda_2) \left(1 - \norm{Q_1\xi}^2 \right).
\end{align*}  Rearranging this inequality gives our desired inequality. \end{proof}

\begin{lemma}\label{lem:expectation-estimate}
Let $\psi\in \bb C^{d_A}\otimes \bb C^{d_B}$ be a unit vector and let $X_1,X_2\in \bb M_{d_A}$ and $Y_1,Y_2\in \bb M_{d_B}$. Let $\rho_A,\rho_B$ be the reduced density matrices of the density matrix $\rho = \psi\psi^*$. Then 
\begin{align}
|\inner{(X_1\otimes Y_1 - X_2\otimes Y_2)\psi}{\psi}| &\leq \|X_1-X_2\|_{\rho_A}\|Y_2^*\|_{\rho_B} + \|Y_1-Y_2\|_{\rho_B}\|X_1^*\|_{\rho_A}, \label{exp-est-1}
\intertext{and}
\|(X_1\otimes Y_1 - X_2\otimes Y_2)\psi\| &\leq \|X_1-X_2\|_{\rho_A} \|Y_2\| + \|Y_1-Y_2\|_{\rho_B}\|X_1\|. \label{exp-est-2}
\end{align}
\end{lemma}

\begin{proof}
Using the triangle inequality and Cauchy--Schwarz inequality, \begin{align*}
\abs{\inner{(X_1\otimes Y_1 - X_2\otimes Y_2)\psi}{\psi}} &\leq \abs{\inner{(I_{d_A}\otimes (Y_1-Y_2))\psi}{(X_1^*\otimes I_{d_B})\psi}}  \\
&\qquad + \abs{\inner{((X_1-X_2)\otimes I_{d_B})\psi}{(I_{d_A}\otimes Y_2^*)\psi}} \\
&\leq \|(I_{d_A}\otimes (Y_1-Y_2))\psi\|\|(X_1^*\otimes I_{d_B})\psi\| \\
&\qquad + \|((X_1-X_2)\otimes I_{d_B})\psi\|\|(I_{d_A}\otimes Y_2^*)\psi\| \\
&= \|Y_1-Y_2\|_{\rho_B}\|X_1^*\|_{\rho_A} + \|X_1-X_2\|_{\rho_A}\|Y_2^*\|_{\rho_B},
\end{align*} where we used identities \eqref{partial-trace-1} and \eqref{partial-trace-2} in the last equation.

The second claimed inequality is derived as follows: \begin{align*}
\|(X_1\otimes Y_1 - X_2\otimes Y_2)\psi\| &\leq \|(X_1\otimes (Y_1-Y_2))\psi\| + \|((X_1-X_2)\otimes Y_2)\psi\| \\
&\leq \|X_1\otimes I_{d_B}\|\|(I_{d_A}\otimes (Y_1-Y_2))\psi\| \\
&\qquad + \|I_{d_A}\otimes Y_2\|\|((X_1-X_2)\otimes I_{d_B})\psi\| \\
&= \|X_1\| \|Y_1-Y_2\|_{\rho_B} + \|Y_2\| \|X_1-X_2\|_{\rho_A},
\end{align*} where the first inequality follows from the triangle inequality, and the second follows from the definition of the operator norm. \end{proof}

\begin{lemma}\label{lem:near-vectors}
Let $\xi,\eta\in \bb C^d$ and let $P\in \bb M_d$ be a projection. Suppose $\left|\norm{\xi}^2 - \norm{\eta}^2 \right|\leq \epsilon_1$ and $\norm{\xi - P\eta} \leq \epsilon_2$ for some $\epsilon_1,\epsilon_2\geq 0$. Then \begin{align}
\norm{\xi - \eta} \leq \epsilon_2 + \sqrt{\epsilon_1 + \left(\norm{\xi}+\norm{\eta}\right)\epsilon_2}.
\end{align}
\end{lemma}

\begin{proof}
Let $P^{\perp} = I_d - P$. Then $\norm{\eta}^2 = \norm{P\eta}^2 + \norm{P^{\perp}\eta}^2$, and thus \begin{align*}
\norm{P^{\perp}\eta}^2 &= \norm{\eta}^2 - \norm{P\eta}^2 \leq \epsilon_1 + \norm{\xi}^2 - \norm{P\eta}^2 =  \epsilon_1 + (\norm{\xi} + \norm{P\eta})(\norm{\xi} - \norm{P\eta}) \\
&\leq \epsilon_1 + (\norm{\xi} + \norm{P\eta})\norm{\xi - P\eta} \leq \epsilon_1 + (\norm{\xi} + \norm{\eta})\epsilon_2.
\end{align*} Thus, $\norm{P^{\perp}\eta} \leq \sqrt{\epsilon_1 + (\norm{\xi} + \norm{P\eta})\epsilon_2},$ using which in $\norm{\xi - \eta} \leq \norm{\xi - P\eta} + \norm{P^{\perp}\eta}$ proves the lemma. \end{proof}

Finally, we show that the correlation $\widetilde{p}_{n, x} \in C_q^s(n,2)$ as in Definition \ref{def:def-of-p*} robustly self-tests the strategy $\widetilde{\mathscr S}$ mentioned therein.

\begin{theorem}[Main result]\label{thm:robust-self-test}
Let $n\geq 3$ and $x\in \Lambda_n$. The synchronous correlation $\widetilde{p}_{n,x} \in C_q^s(n,2)$ induced by the strategy $\widetilde{\mathscr S}$ in Definition \ref{def:def-of-p*} self-tests $\widetilde{\mathscr S}$ robustly.
\end{theorem}

\begin{proof}
Theorem \ref{thm:self-test} states that $\widetilde{p}_{n, x}$ self-tests $\widetilde{\mathscr S}$.  We must show that for all $\epsilon > 0$, there is a $\delta > 0$ such that if $p \in C_q(n, 2)$ is a correlation induced by a strategy $\mathscr S$ and $\| p -\widetilde{p}_{n, x} \|_1 \leq \delta$, then $\widetilde{\mathscr S}$ is a local $\epsilon$-dilation of $\mathscr S$.  Fix $\epsilon > 0$.

With the notation of Definition \ref{def:def-of-p*}, let \begin{align*}
\widetilde{N} = \sum_{v=1}^n \widetilde{P}_v \otimes \widetilde{P}_v^T \in \bb M_d \otimes \bb M_d
\end{align*} and note that $\widetilde{N}$ is positive.  By Lemma \ref{lem:max-evalue}, $x$ is the largest eigenvalue of $\widetilde{N}$, and the eigenspace corresponding to $x$ is the one-dimensional space spanned by the maximally entangled state $\varphi_d$.  In particular, $\widetilde{N}$ has more than one eigenvalue---let $\lambda_2 \geq 0$ be the second largest eigenvalue of $\widetilde{N}$.

Let $\epsilon' > 0$ be such that $\epsilon' < (x - \lambda_2)/(2n + 1)$ and 
\begin{align}\label{epsilon-prime}
  2\epsilon' + \beta + \sqrt{5 \epsilon' + 4 \sqrt{\epsilon'} + 2 \beta} \leq \epsilon \qquad \text{where} \qquad \beta \coloneqq \sqrt{\frac{2(2n+1)\epsilon'}{x - \lambda_2}}.
\end{align} Apply Theorem \ref{thm:Gowers-Hatami-N} to produce $m \in \mathbb N$ and $\delta' > 0$ such that Condition \eqref{GH-condition} holds with $\delta'$ and $\epsilon'$ in place of $\delta$ and $\epsilon$.  Then define
\begin{align}\label{delta}
  \delta \coloneqq \min\left\{ \epsilon', \left(\frac{\delta'}{n+1}\right)^4, \left(\frac{\delta'}{2m}\right)^2 \right\} > 0.
\end{align}

Let $p \in C_q(n, 2)$ be a correlation with $\|p - \widetilde{p}_{n, x}\|_1 \leq \delta$ and let \begin{align}
\mathscr{S} = \big(\psi\in \bb C^{d_A}\otimes \bb C^{d_B}, \{E_v, I_{d_A} - E_v\}_{v=1}^n, \{F_w, I_{d_B} - F_w\}_{w=1}^n\big)
\end{align} be a strategy inducing $p$.  We will show $\widetilde{\mathscr S}$ is a local $\epsilon$-dilation of $\mathscr S$.

Let $\rho = \psi \psi^* \in \bb M_{d_A} \otimes \bb M_{d_B}$ and let $\rho_A \in \bb M_{d_A}$ and $\rho_B \in \bb M_{d_B}$ be the reduced density matrices of $\rho$.  By Lemma \ref{thm:almost-corr-nonsyn}, $(E_1, \ldots, E_n)$ is a $(C\delta^\frac14, \rho_A)$-representation of the relations given in \eqref{eq:relations}, where
\[ C \coloneqq \sqrt{n^2 + (1+2x)\sqrt{\delta}}. \]
As $x \leq n$, we have $\delta \leq \epsilon' < (x - \lambda_2)/(2n + 1) < \frac12$.  Now, since $x \leq n$ and $\delta \leq 1$, we have
\[ C \leq \sqrt{n^2 + 2n + 1} = n + 1. \]
Therefore, $C \delta^\frac14 \leq \delta'$ by the choice of $\delta$.  Hence $(E_1, \ldots, E_n)$ is a $(\delta', \rho_A)$-representation of the relations given in \eqref{eq:relations}.  Therefore, (a) and (b) of Condition \ref{GH-condition} hold with $(\delta', \rho_A, d_A)$ in place of $(\delta, \rho, r)$.  Since $2m\sqrt{\delta} \leq \delta'$, Lemma \ref{lem:st-app-tra} implies that part (c) of Condition \ref{GH-condition} also holds.

The conclusion of Condition \ref{GH-condition} provides $r_A \in \mathbb N$ and an isometry $V_A \colon \bb C^{d_A} \rightarrow \bb C^d \otimes \bb C^{r_A}$ with
\begin{equation}\label{compression-1}
  \| E_v - V_A^* (\widetilde{P}_v \otimes I_{r_A} )V_A \|_{\rho_A}\leq \epsilon' \qquad \text{for all $v \in [n]$}.
\end{equation}
Similarly, there are $r_B \in \mathbb N$ and an isometry $V_B \colon \bb C^{d_B} \rightarrow \bb C^d \otimes \bb C^{r_B}$ such that
\begin{equation}\label{compression-2}
  \| F_v - V_B^* (\widetilde{P}_v^T \otimes I_{r_B}) V_B \|_{\rho_B} \leq \epsilon' \qquad \text{for all $v \in [n]$}.
\end{equation}
We work towards constructing a quantum state $\psi_\mathrm{junk} \in \bb C^{r_A} \otimes \bb C^{r_B}$ such that
\begin{equation}\label{junk-state-estimate}
  (V_A \otimes V_B) \psi \approx_\epsilon \varphi_d \otimes \psi_\mathrm{junk}.
\end{equation}

Note that Inequalities \eqref{compression-1} and \eqref{compression-2} together with Lemma \ref{lem:expectation-estimate} imply
\begin{align}
\left| \inner{\big(E_v\otimes F_w - (V_A\otimes V_B)^*(\widetilde{P}_v\otimes\widetilde{P}_w^T\otimes I_{r_A}\otimes I_{r_B})(V_A\otimes V_B)\big)\psi}{\psi}\right| \notag\\
\leq \|E_v - V_A^*(\widetilde{P}_v\otimes I_{r_A})V_A\|_{\rho_A} + \|F_w - V_B^*(\widetilde{P}_w^T\otimes I_{r_B})V_B\|_{\rho_B} \leq 2\epsilon'. \label{another-estimate}
\end{align}
Let $N = \sum_{v = 1}^n E_v \otimes F_v \in \bb M_{d_A} \otimes \bb M_{d_B}$.  Taking $v = w$ in Inequality \eqref{another-estimate} and summing over $v \in [n]$ yields
\begin{equation}\label{close-to-N-psi-psi}
\left| \inner{N \psi}{\psi} - \inner{(\widetilde{N} \otimes I_{r_A} \otimes I_{r_B})(V_A \otimes V_B)\big) \psi}{(V_A \otimes V_B) \psi} \right| \leq 2n \epsilon'.
\end{equation}
Note that $\inner{N\psi}{\psi} = \sum_{v=1}^n p(1, 1 |v, v)$.  Since $\|p - \widetilde{p}_{n, x}\|_1 \leq \delta \leq \epsilon'$ and $\sum_{v=1}^n \widetilde{p}_{n, x}(1, 1|v, v) = x$, Inequality \ref{close-to-N-psi-psi} implies
\begin{equation}\label{lower-bound}
   \inner{(\widetilde{N} \otimes I_{r_A} \otimes I_{r_B})(V_A \otimes V_B)\big) \psi}{(V_A \otimes V_B) \psi} \geq x - (2n+1)\epsilon'.
\end{equation}

Recall that by Lemma \ref{lem:max-evalue}, $x$ is the largest eigenvalue of $\widetilde{N}$, and hence also of $\widetilde{N} \otimes I_{r_A} \otimes I_{r_B}$.  Let $Q$ be the projection on to the eigenspace of $\widetilde{N} \otimes I_{r_A} \otimes I_{r_B}$ corresponding to $x$.  Equation \eqref{lower-bound} and Lemma \ref{lem:approx-evector} yield
\begin{equation}\label{alpha-lower-bound}
  \alpha\coloneqq \|Q (V_A \otimes V_B) \psi\| \geq \left(1 - \frac{(2n+1) \epsilon'}{x - \lambda_2} \right)^{1/2}> 0,
\end{equation}
where the last inequality holds by the choice of $\epsilon'$.  Then $\alpha^{-1} Q (V_A \otimes V_B) \psi$ is a unit vector in the eigenspace of $\widetilde{N} \otimes I_{r_A} \otimes I_{r_B}$ corresponding to $x$.  As the eigenspace of $\widetilde{N}$ is spanned by $\varphi_d$, there is a unique quantum state $\psi_\mathrm{junk} \in \bb C^{r_A} \otimes \bb C^{r_B}$ such that $Q(V_A \otimes V_B) \psi = \alpha \varphi_d \otimes \psi_\mathrm{junk}$.

To simplify notation, let $\psi' = (V_A \otimes V_B) \psi$, and note that
\begin{align*}
  \big\| \psi' - \|Q\psi'\|^{-1} Q\psi' \big\|^2 &= \|(1 - Q) \psi'\|^2 + \big|1 - \|Q \psi'\|^{-1}\big|^2 \|Q \psi'\|^2 \\
  &= 1 - \|Q \psi'\|^2 + (\|Q\psi'\|^{-1} - 1)^2 \|Q\psi'\|^2 \\
  &= 1 - \|Q \psi'\|^2 + 1 - 2 \|Q \psi'\| + \|Q\psi'\|^2 \\
  &= 2 (1 - \|Q \psi'\|).
\end{align*} Therefore,
\[ \| (V_A \otimes V_B) \psi - \varphi_d \otimes \psi_\mathrm{junk} \| \leq \sqrt{2 (1 - \alpha)} \leq \sqrt{2(1 - \alpha^2)}. \]
Combining this inequality with Inequality \eqref{alpha-lower-bound},
\begin{equation}\label{junk-state-estimate-2}
  \| (V_A \otimes V_B) \psi - \varphi_d \otimes \psi_\mathrm{junk} \| \leq \sqrt{\frac{2(2n+1)\epsilon'}{x - \lambda_2}} = \beta.
\end{equation}
In particular, Estimate \eqref{junk-state-estimate} holds.

We now work to show
\begin{align}\label{key-estimate}
(V_A \otimes V_B) (E_v \otimes F_w) \psi \approx_\epsilon \big((\widetilde{P}_v \otimes \widetilde{P}_w^T) \varphi_d \big) \otimes \psi_\mathrm{junk}  
\end{align}
for all $v, w \in [n]$.  Fix $v, w \in [n]$ and, to declutter the notation, define
\[  V \coloneqq V_A \otimes V_B \qquad \text{and} \qquad P_{v, w} \coloneqq \widetilde{P}_v \otimes \widetilde{P}_w^T \otimes I_{r_A} \otimes I_{r_B}. \]

By Inequalities \eqref{exp-est-2}, \eqref{compression-1}, and \eqref{compression-2}, we have
\[ \|V^* P_{v, w} V \psi - (E_v \otimes F_w)\psi \| \leq \|V_A^* (\widetilde{P}_v \otimes I_{r_A}) V_A - E_v\|_{\rho_A} + \|V_B^* (\widetilde{P}_w^T \otimes I_{r_B}) V_B - F_w\|_{\rho_B} \leq 2 \epsilon'. \]
Since $V$ is an isometry, this implies
\[ \|VV^* P_{v, w} V \psi - V(E_v \otimes F_w)\psi \| \leq 2 \epsilon'. \]
Combining this with Inequality \eqref{junk-state-estimate-2} and the triangle inequality, we have
\begin{equation}\label{key-est-part-1}
\|V(E_v \otimes F_w) \psi - VV^*P_{v, w} (\varphi_d \otimes \psi_\mathrm{junk})\| \leq 2 \epsilon' + \beta.
\end{equation}

Note that
\begin{equation}\label{norm-computation-1}
\| P_{v, w}(\varphi_d \otimes \psi_\mathrm{junk})\|^2 = \inner{\left(\widetilde{P}_v \otimes \widetilde{P}_w^T\right)\varphi_d}{\varphi_d} = \widetilde{p}_{n, x}(1, 1|v, v).
\end{equation}
Further, since $p(1, 1|v, w) = \inner{(E_v \otimes F_w) \psi}{\psi}$, Inequality \eqref{exp-est-1} and parts (d) and (e) of Lemma \ref{lem:approx-vectors} imply
\begin{equation}\label{norm-computation-2}
  \left|\| (E_v \otimes F_w) \psi\|^2 - p(1, 1 |v, w)\right| \leq \|E_v - E_v^2\|_{\rho_A} + \|F_w - F_w^2\|_{\rho_B} \leq 4 \sqrt{\delta} \leq 4 \sqrt{\epsilon'}. 
\end{equation}
Since $V$ is an isometry and $\|p - \widetilde{p}_{n, x}\|_1 \leq \delta \leq \epsilon'$, Equation \eqref{norm-computation-1} and Inequality \eqref{norm-computation-2} imply
\begin{equation}\label{key-est-part-2}
  \left| \|V(E_v \otimes F_w)\psi\|^2 - \|P_{v, w}(\varphi_d \otimes \psi_\mathrm{junk}) \|^2 \right| \leq \epsilon' + 4 \sqrt{\epsilon'}.
\end{equation}

As $VV^*$ is a projection, Inequalities \eqref{key-est-part-1} and \eqref{key-est-part-2} and Lemma \ref{lem:near-vectors} imply
\[ \| V(E_v \otimes F_w) \psi - P_{v, w} (\varphi_d \otimes \psi_\mathrm{junk}) \| \leq 2 \epsilon' + \beta + \sqrt{5 \epsilon' + 4 \sqrt{\epsilon'} + 2 \beta} \leq \epsilon, \]
and this proves Estimate \eqref{key-estimate}.  Similarly, for each $v, w \in [n]$, using the estimates
\begin{align*}
  \big\| (I_{d_A} - E_v) - V_A^* \big( (I_{d_A} - \widetilde{P}_v) \otimes I_{r_A}\big) V_A \big\|_{\rho_A} &\leq \epsilon' \quad \text{and} \\
  \big\|(I_{d_B} - F_w) - V_B^*\big((I_{d_B} - \widetilde{P}_w^T) \otimes I_{r_B}\big) V_B \big\|_{\rho_B} &\leq \epsilon',
\end{align*}
which follow from Inequalities \eqref{compression-1} and \eqref{compression-2}, one can show
\begin{align*}
  (V_A \otimes V_B)\big(E_v \otimes (I_{d_B} - F_w)\big) \psi &\approx_\epsilon \Big(\big(\widetilde{P}_v \otimes (I_d - \widetilde{P}_w^T)\big) \varphi_d \Big) \otimes \psi_\mathrm{junk}, \\
  (V_A \otimes V_B)\big((I_{d_A} - E_v) \otimes F_w\big) \psi &\approx_\epsilon \Big(\big((I_d - \widetilde{P}_v) \otimes \widetilde{P}_w^T\big) \varphi_d \Big) \otimes \psi_\mathrm{junk}, \text{ and} \\
  (V_A \otimes V_B)\big((I_{d_A} - E_v) \otimes (I_{d_B} - F_w)\big) \psi &\approx_\epsilon \Big(\big((I_d - \widetilde{P}_v) \otimes (I_d - \widetilde{P}_w^T)\big)\varphi_d\big) \otimes \psi_\mathrm{junk}.
\end{align*}
These last three estimates together with Estimates \eqref{junk-state-estimate} and \eqref{key-estimate} show that $\widetilde{\mathscr{S}}$ is a local $\epsilon$-dilation of $\mathscr{S}$.
\end{proof}

\section{Implications}\label{sec:implications}
We discuss some corollaries which result from the specific case of four projections adding up to scalar times the identity.

\begin{cor}\label{cor:n=4}
Let $n = 4$. For any $k\in \bb N$, there exists four rank $k$ projections $\widetilde P_{k,1},\widetilde P_{k,2}, \widetilde P_{k,3}, \widetilde P_{k,4}$ in $\bb M_{2k+1}$ such that $\widetilde P_{k,1} + \widetilde P_{k,2} + \widetilde P_{k,3} + \widetilde P_{k,4} = \frac{4k}{2k+1}I_{2k+1}$. Each of the following quantum strategies  \begin{align*}
\widetilde{\mathscr{S}}_k = (\varphi_{2k+1}, \{ \widetilde P_{k,v}, I_{2k+1} - \widetilde P_{k,v}\}_{v=1}^4, \{ \widetilde P_{k,w}^T, I_{2k+1} - \widetilde P_{k,w}^T \}_{w=1}^4),
\end{align*} can be robustly self-tested from the correlations $\widetilde{p}_{4,k}$ that each strategy induces.
\end{cor}

\begin{proof}
By the recurrence relation given in Theorem \ref{thm:KRS} it is readily computed that $\Lambda_4 = \left\lbrace \frac{4k}{2k+1}\right\rbrace_{k = 0}^{\infty}$. Observe that each fraction $\frac{4k}{2k+1}$ is already in lowest terms. Hence by Theorem \ref{thm:KRS} there exist four projections $\widetilde P_{k,1},\widetilde P_{k,2}, \widetilde P_{k,3}, \widetilde P_{k,4}$ in $\bb M_{2k+1}$ such that $\widetilde P_{k,1} + \widetilde P_{k,2} + \widetilde P_{k,3} + \widetilde P_{k,4} = \frac{4k}{2k+1}I_{2k+1}$. By Lemma \ref{prop:const-trace}, we observe that each of these projections is of rank $k$. Rest of the corollary then follows from Theorem \ref{thm:robust-self-test}.
\end{proof}

The following two observations immediately follow from Corollary \ref{cor:n=4}.

\begin{cor}\label{cor:robust-self-test-max-ent}
The maximally entangled state $\varphi_d$ in each odd dimension $d \geq 3$ can be robustly self-tested by quantum correlations with four inputs and two outputs. 
\end{cor}

\begin{cor}
Given any natural number $k$ there exist four projections of rank $k$ which can be robustly self-tested by quantum correlations with four inputs and two outputs.
\end{cor}

	In general, a quantum state is represented by a density matrix $\rho\in \bb M_d$. This reduces to the usual vector state formulation if $\rho$ is rank one. Thus, in general, a description of a quantum strategy is given by \begin{align*}
		\mathscr{S} = \left(\rho\in \bb M_{d_A}\otimes \bb M_{d_B},\{E_{v,i}:v\in [n_A], i\in [k_A]\},\{F_{w,j}:w\in [k_B], j\in [k_B]\} \right),
	\end{align*} and its induced quantum correlation is given by $p(i,j|v,w) = \tr((E_{v,i}\otimes F_{w,j})\rho)$. One can show that this quantum correlation $p\in C_q(n_A,n_B,k_A,k_B)$ induced by $\mathscr{S}$ can also be obtained by a quantum strategy $\mathscr S'$ where the quantum state is now given by a unit vector by using the notion of \emph{purification} \cite[Section 2.5]{NC}.
	
	One can formulate a definition of self-testing which takes this more general view of quantum states in to account. Indeed, this is how self-testing is defined in \cite[Definition C.1]{Goh18}. It is then proved there \cite[Proposition C.1]{Goh18} that if a quantum correlation $\widetilde p\in C_q(n_A,n_B,k_A,k_B)$ self-tests a quantum strategy, then $\widetilde{p}$ must be an extreme point of $C_q(n_A,n_B,k_A,k_B)$. 
	
	While in this article we don't work\footnote{Also, we do not know whether Definition \ref{def:self-test} is equivalent to the definition given in \cite{Goh18}.} with the definition in \cite{Goh18}, one can show that all of our results can be suitably generalized. More precisely, let $n\geq 3$ and $x\in \Lambda_n$ with $x = \frac{b}{d}$ in lowest terms. Let $\widetilde{p}_{n,x}\in C_q^s(n,2)$ and $\widetilde{\mathscr{S}}$ be as in Definition \ref{def:def-of-p*}. One can show that $\widetilde{p}_{n,x}$ self-tests the strategy $\widetilde{\mathscr{S}}$ in the following sense as well. Consider a quantum strategy $(\rho\in \bb M_{d_A}\otimes \bb M_{d_B},\{E_v,I_{d_A} - E_v\}_{v=1}^n, \{F_w, I_{d_B} - F_w\}_{w=1}^n)$ which also induces $\widetilde{p}_{n,x}$ and where $\rho$ is a density matrix. Then, there exist isometries $V_A\colon \bb C^{d_A} \to \bb C^{d} \otimes \cl K_A$ and $V_B\colon \bb C^{d_B} \to \bb C^{d} \otimes \cl K_B$ for some finite-dimensional Hilbert spaces $\cl K_A,\cl K_B$, and a density matrix $\rho_{\mathrm{junk}}\in \mathfrak{B}(\cl K_A\otimes \cl K_B)$ such that for all $v,w$, we have \begin{align}
		(V_A\otimes V_B)((E_v\otimes F_w)\rho)(V_A\otimes V_B)^* = ((\widetilde{P}_{v}\otimes \widetilde{P}_{w}^T)\widetilde{\rho}) \otimes \rho_{\mathrm{junk}},
	\end{align} where $\widetilde{\rho} = \varphi_d\varphi_d^*$. With this result, it follows that for all $n\geq 3$ and $x\in \Lambda_n$, the quantum correlation $\widetilde{p}_{n,x}$ is an extreme point of $C_q(n,2)$.
	
\paragraph{\bf Acknowledgements.}  L.~M. and J.~P. are supported by the Villum Fonden via the QMATH Centre of Excellence (Grant No.~10059). C.~S. is partially supported by NSF Grant DMS-2000129. We thank the QIP 2021 referees for their feedback and Honghao Fu for clarifying his results in \cite{Fu19}.

\appendix
\section{A short introduction to C*-algebras}\label{appendixA}
In this appendix we give some overview of C$^*$-algebras that we have used in this article. Some general references on this subject are \cite{Murphy,CBMOA}.

A \emph{C$^*$-algebra} $\mathfrak{A}$ is a subset of the set of bounded linear operators $\bh$ for some Hilbert space $\cl H$ satisfying the following properties: \begin{enumerate}
\item $\mathfrak{A}$ is a unital (complex) algebra, that is, for $S,T\in \mathfrak{A}$ and $\lambda\in \bb C$, we have $S+\lambda T\in \mathfrak{A}$ and $\lambda S$, and there is an element $1_{\mathfrak{A}}\in \mathfrak{A}$ which serves as the unit element for $\mathfrak{A}$,
\item $\mathfrak{A}$ is $*$-closed, that is, if $S\in \mathfrak{A}$, then its adjoint $S^*\in \mathfrak{A}$ as well, and,
\item $\mathfrak{A}$ is closed in the operator norm.
\end{enumerate} Sometimes, such a C$^*$-algebra is also called a \emph{concrete C$^*$-algebra} or a \emph{C$^*$-algebra of operators} to distinguish it from an abstract C$^*$-algebra which we describe next.

Let $\mathfrak{A}$ be a unital Banach algebra, that is, $\mathfrak{A}$ is an algebra equipped with a norm $\norm{.}$ satisfying $\norm{ab}\leq \norm{a}\norm{b}$ for all $a,b\in \mathfrak{A}$, and $\mathfrak{A}$ is complete (every Cauchy sequence converges). Assume that $\mathfrak{A}$ has an involution operation $*:\mathfrak{A}\to \mathfrak{A}$ which satisfies $1_{\mathfrak{A}}^* = 1_{\mathfrak{A}}$ and \begin{align*}
(a+\lambda b)^* = a^* + \overline{\lambda}b^*, \qquad (ab)^* = b^*a^*, \qquad \qquad (a^*)^* = a,
\end{align*} for all $a,b\in \mathfrak{A}$ and $\lambda \in \bb C$. We say that $\mathfrak{A}$ is an \emph{abstract C$^*$-algebra} if $\norm{a^*a} = \norm{a}^2$ for all $a\in \mathfrak{A}$. 

A linear map $\pi:\mathfrak{A} \to \mathfrak{B}$ between two abstract C$^*$-algebras is called a \emph{$*$-homomorphism} if it is unital $\pi(1_{\mathfrak{A}}) = 1_{\mathfrak{B}}$, and $\pi(ab) = \pi(a)\pi(b)$ and $\pi(a^*) = \pi(a)^*$ for all $a,b\in \mathfrak{A}$.

Observe that $\bh$ with the operator norm is an abstract C$^*$-algebra and hence every concrete C$^*$-algebra is an abstract C$^*$-algebra. Conversely, every abstract C$^*$-algebra can be identified with a concrete C$^*$-algebra on some Hilbert space: 

\begin{theorem}[Gelfand-Naimark-Segal]
Let $\mathfrak{A}$ be an abstract C$^*$-algebra. Then, there is some Hilbert space $\cl H$ and an isometric $*$-homomorphism $\pi:\mathfrak{A}\to \bh$.
\end{theorem}

Thus, one may identify an abstract C$^*$-algebra $\mathfrak{A}$ with the concrete C$^*$-algebra $\mathfrak{B}:= \pi(\mathfrak{A})$. A $*$-homomorphism from a C$^*$-algebra $\mathfrak{A}$ into some $\bh$ is called a \emph{representation} of $\mathfrak{A}$.

A \emph{state} on a C$^*$-algebra $\mathfrak{A}$ is a linear functional $\phi:\mathfrak{A}\to \bb C$ such that $\phi(a^*a) \geq 0$ for all $a\in \mathfrak{A}$ and $\phi(1_{\mathfrak{A}}) = 1$. A state $\phi$ is called \emph{tracial} if $\phi(ab) = \phi(ba)$ for all $a,b\in \mathfrak{A}$. We let $S(\mathfrak{A})$ denote the set of all states on $\mathfrak{A}$. We say that a net of states $(\phi_{\lambda})_{\lambda\in \Lambda}$ converges to a state $\phi$ if $\lim_{\lambda} |\phi_{\lambda}(a) - \phi(a)| = 0$ for all $a\in \mathfrak{A}$, that is, if $\phi_{\lambda}$ converges pointwise to $\phi$. This gives rise to a topology on $S(\mathfrak{A})$ called the \emph{weak$^*$-topology}, with respect to which $S(\mathfrak{A})$ is compact.

Let $\pi_1:\mathfrak{A}\to \mathfrak{B}(\cl H_1)$ and $\pi_2:\mathfrak{A}\to \mathfrak{B}(\cl H_2)$ be two representations of a C$^*$-algebra $\mathfrak{A}$. We say $\pi_1$ and $\pi_2$ are \emph{unitarily equivalent} if there is some unitary $U:\cl H_1 \to \cl H_2$ such that $\pi_2(a) = U\pi_1(a)U^*$ for all $a\in \mathfrak{A}$. A representation $\pi:\mathfrak{A}\to \bh$ is said to be \emph{irreducible} if there is no invariant closed subspace of $\pi(\mathfrak{A})$ apart from $\{0\}$ and $\cl H$.

Give a non-empty set $A\subseteq \bh$, we define the \emph{C$^*$-algebra generated by $A$}, denoted $C^*(A)$, to be the smallest $C^*$-algebra containing $A$.

Let $G = \{a_1,\dots,a_n\}$ be a non-empty set of \emph{generators}, and let $R = \{p_1,\dots,p_m\}$ be a finite set of \emph{relations} where $p_i$ is a polynomial in $2n$ noncommuting variables $a_1,\dots,a_n,a_1^*,\dots,a_n^*$.  Let $\mathfrak{A}$ be the free $*$-algebra generated by $G$. Then an $n$-tuple $(T_1,\dots,T_n)$ in some $\bh$ satsfiying the relations gives rise to a $*$-representation of $\mathfrak{A}$. For $a\in \mathfrak{A}$, define $\norm{a} = \sup\{\norm{\pi(a)}: \pi \text{ is a *-representation of } (G,R)\}$. If $\|a_i\| < \infty$ for all $i \in [n]$, then $\|\cdot\|$ is a C$^*$-seminorm on $\mathfrak A$, that is, $\norm{a^*a} = \norm{a}^2$ for all $a\in \mathfrak{A}$. Letting $N = \{a\in \mathfrak{A}:\norm{a} = 0\}$, the completion of $\mathfrak A/N$ with respect to the induced norm is called the \emph{universal C$^*$-algebra} on $(G,R)$, denoted by $C^*(G,R)$. The universal property says that whenever there is an $n$-tuple $(A_1,\dots,A_n)$ of operators in $\bh$ for some Hilbert space $\cl H$ satisfying the relations $R$, there is a representation $\pi:C^*(G,R)\to \bh$ with $\pi(a_i) = A_i$ for all $i\in [n]$. For more information on universal C$^*$-algebras we refer the reader to \cite[Section II.8.3]{Black}.

\section{A recipe to generate four projections adding up to scalars times identity}\label{appendixB}

For each $k\in \bb N$, we present a recipe to generate four projections $\widetilde{P}_{k,1}, \dots, \widetilde{P}_{k,4}$ such that $\widetilde{P}_{k,1} + \dots + \widetilde{P}_{k,4} = \frac{4k}{2k+1}$ (see Corollary \ref{cor:n=4}). This is a slightly simplified presentation from~\cite{KRS}. We follow the following iterative procedure.

We begin with four unit vectors $\xi_1,\xi_2,\xi_3,\xi_4\in \mathbb{R}^3$ which form the vertices of a regular tetrahedron. These vectors satisfy $\inner{\xi_i}{\xi_j}=\frac{-1}{3}$ for all $i\neq j$, and are given by  \begin{align*}
	\xi_1 = \begin{bmatrix}
		1 \\ 0 \\ 0
	\end{bmatrix}, \quad \xi_2 = \begin{bmatrix}
		\frac{-1}{3} \\ \frac{2\sqrt{2}}{3} \\ 0
	\end{bmatrix}, \quad \xi_3 = \begin{bmatrix}
		\frac{-1}{3} \\ \frac{-\sqrt{2}}{3} \\ \sqrt{\frac{2}{3}}
	\end{bmatrix}, \quad \xi_4 = \begin{bmatrix}
		\frac{-1}{3} \\ \frac{-\sqrt{2}}{3} \\ -\sqrt{\frac{2}{3}}
	\end{bmatrix}.
\end{align*} Set $\widetilde{P}_{1,v} = \xi_v\xi_v^*$ to be the rank-1 projection on the subspace spanned by $\xi_v$, for $v\in [4]$. Then \begin{align*}
	\widetilde{P}_{1,1} = \begin{bmatrix}
		1 & 0 & 0 \\
		0 & 0 & 0 \\
		0 & 0 & 0 
	\end{bmatrix}, \quad
	\widetilde{P}_{1,2} = \begin{bmatrix}
		\frac{1}{9} & \frac{-2\sqrt{2}}{9} & 0 \\
		\frac{-2\sqrt{2}}{9} & \frac{8}{9} & 0 \\
		0 & 0 & 0
	\end{bmatrix}, \quad
	\widetilde{P}_{1,3} = \begin{bmatrix}
		\frac{1}{9} & \frac{\sqrt{2}}{9} & \frac{-\sqrt{2}}{3\sqrt{3}} \\
		\frac{\sqrt{2}}{9} & \frac{2}{9} & \frac{-2}{3\sqrt{3}} \\
		\frac{-\sqrt{2}}{3\sqrt{3}} & \frac{-2}{3\sqrt{3}} & \frac{2}{3}
	\end{bmatrix}, \quad
	\widetilde{P}_{1,4} = \begin{bmatrix}
		\frac{1}{9} & \frac{\sqrt{2}}{9} & \frac{\sqrt{2}}{3\sqrt{3}} \\
		\frac{\sqrt{2}}{9} & \frac{2}{9} & \frac{2}{3\sqrt{3}} \\
		\frac{\sqrt{2}}{3\sqrt{3}} & \frac{2}{3\sqrt{3}} & \frac{2}{3}
	\end{bmatrix}.
\end{align*} One readily checks that $\widetilde{P}_{1,1}+\dots+\widetilde{P}_{1,4} = \frac{4}{3}I_3$. 

More generally, for $k\geq 1$, we show how to construct projections $\widetilde{P}_{k+1,1},\dots, \widetilde{P}_{k+1,4}$ which sum up to $\frac{4(k+1)}{2(k+1)+1}$ from the set of projections $\widetilde{P}_{k,1},\dots, \widetilde{P}_{k,4}$. We begin with projections $\widetilde{P}_{k,1}, \dots, \widetilde{P}_{k,4}$ satisfying $\widetilde{P}_{k,1}+\dots+\widetilde{P}_{k,1} = \frac{4k}{2k+1}I_{2k+1}$. Set $\widetilde{Q}_{k,v} = I_{2k+1} - \widetilde{P}_{k,v}$. Observe that $\mathrm{rank}(\widetilde{Q}_{k,v}) = k+1$ since $\mathrm{rank}(\widetilde{P}_{k,v}) = k$ (Corollary \ref{prop:const-trace}). Thus if $\widetilde{Q}_{k,v} = U_{k,v}DV_{k,v}^*$ is the singular decomposition with $D = I_{k+1}\oplus 0_{k}$, we pick the first $k+1$ columns from each of $U_{k,v}$ and join them laterally to form a larger matrix $\Gamma_1$ of size $(2k+1) \times 4(k+1)$: \begin{align*}
	\Gamma_1 = \begin{bmatrix} \text{first } k+1 \text{ columns of } U_{k,1}  & \dots & \text{first } k+1 \text{ columns of } U_{k,4} \end{bmatrix}
\end{align*} The null space of $\Gamma_1$ is then spanned by an orthonormal basis $\{\eta_1,\dots,\eta_{2k+3}\}$, and therefore \begin{align*}
	\Gamma_2 := \begin{bmatrix}
		\eta_1^T \\ \vdots \\ \eta_{2k+3}^T
	\end{bmatrix}
\end{align*} is a $(2k+3)\times 4(k+1)$ matrix. Let $R_1$ be the first $k+1$ columns of $\Gamma_2$, $R_2$ be the second $k+1$ columns of $\Gamma_2$, and so on. Finally, set \begin{align*}
	\widetilde{P}_{k+1,v} = \frac{4(k+1)}{2(k+1)+1} R_vR_v^*
\end{align*} for all $v\in [4]$, to get the desired set of projections.

We illustrate the procedure given above by constructing projections $\widetilde{P}_{2,1},\dots,\widetilde{P}_{2,4}$ in $\bb M_5$ such that $\widetilde{P}_{2,1}+\dots+\widetilde{P}_{2,4} = \frac{8}{5}I_5$.

Set $\widetilde{Q}_{1,v} = I_3 - \widetilde{P}_{1,v}$ for each $v\in [4]$. Notice that each $\widetilde{Q}_{1,v}$ is a rank-2 projection and is given by \begin{align*}
	\widetilde{Q}_{1,1} = \begin{bmatrix}
		0 & 0 & 0 \\
		0 & 1 & 0 \\
		0 & 0 & 1
	\end{bmatrix}, \quad
	\widetilde{Q}_{1,2} = \begin{bmatrix}
		\frac{8}{9} & \frac{2\sqrt{2}}{9} & 0 \\
		\frac{2\sqrt{2}}{9} & \frac{1}{9} & 0 \\
		0 & 0 & 1
	\end{bmatrix}, \quad
	\widetilde{Q}_{1,3} = \begin{bmatrix}
		\frac{8}{9} & \frac{-\sqrt{2}}{9} & \frac{\sqrt{2}}{3\sqrt{3}} \\
		\frac{-\sqrt{2}}{9} & \frac{7}{9} & \frac{2}{3\sqrt{3}} \\
		\frac{\sqrt{2}}{3\sqrt{3}} & \frac{2}{3\sqrt{3}} & \frac{1}{3}
	\end{bmatrix}, \quad
	\widetilde{Q}_{1,4} = \begin{bmatrix}
		\frac{8}{9} & \frac{-\sqrt{2}}{9} & \frac{-\sqrt{2}}{3\sqrt{3}} \\
		\frac{-\sqrt{2}}{9} & \frac{7}{9} & \frac{-2}{3\sqrt{3}} \\
		\frac{-\sqrt{2}}{3\sqrt{3}} & \frac{-2}{3\sqrt{3}} & \frac{1}{3}
	\end{bmatrix}.
\end{align*}

For each $v\in [4]$, let $\widetilde{Q}_{1,v} =  U_{1,v}DV^*_{1,v}$ be the singular value decomposition of $\widetilde{Q}_{1,v}$, where $D = \mathrm{diag}(1,1,0)$ and \begin{align*}
U_{1,1} &= \begin{bmatrix} 0 & 0 & 1 \\ 0 & 1 & 0 \\ 1 & 0 & 0 \end{bmatrix}, \qquad
U_{1,2} = \begin{bmatrix} 0 & \frac{2\sqrt{2}}{3} & \frac{-1}{3} \\ 0 & \frac{1}{3} & \frac{2\sqrt{2}}{3} \\ 1 & 0 & 0 \end{bmatrix}, \\
U_{1,3} &= \begin{bmatrix} \sqrt{\frac{6}{7}} & \frac{-\sqrt{2}}{3\sqrt{7}} & \frac{-1}{3} \\ 0 & \frac{\sqrt{7}}{3} & \frac{-\sqrt{2}}{3} \\ \frac{1}{\sqrt{7}} & \frac{2}{\sqrt{21}} & \sqrt{\frac{2}{3}} \end{bmatrix}, \qquad
U_{1,4} = \begin{bmatrix} -\sqrt{\frac{6}{7}} & \frac{-\sqrt{2}}{3\sqrt{7}} & \frac{1}{3} \\ 0 & \frac{\sqrt{7}}{3} & \frac{\sqrt{2}}{3} \\ \frac{1}{\sqrt{7}} & \frac{-2}{\sqrt{21}} & \sqrt{\frac{2}{3}}\end{bmatrix}.
\end{align*} We pick the first two columns of $U_{1,v}$ (which span the column space of $\widetilde{Q}_{1,v}$) and join them laterally to form a larger matrix: 
\begin{align*} 
	\Gamma_1 = \begin{bmatrix}
		0 & 0 & 0 & \frac{2 \sqrt{2}}{3} & \sqrt{\frac{6}{7}} & \frac{-\sqrt{2}}{3\sqrt{7}} & -\sqrt{\frac{6}{7}} & \frac{-\sqrt{2}}{3\sqrt{7}} \\
		0 & 1 & 0 & \frac{1}{3} & 0 & \frac{\sqrt{7}}{3} & 0 & \frac{\sqrt{7}}{3} \\
		1 & 0 & 1 & 0 & \frac{1}{\sqrt{7}} & \frac{2}{\sqrt{21}} & \frac{1}{\sqrt{7}} & -\frac{2}{\sqrt{21}} \\
	\end{bmatrix}.
\end{align*} Let $\{\eta_1,\dots,\eta_5\}$ be an orthonormal basis of the null space of $\Gamma_1$. Form the matrix \begin{align*}
	\Gamma_2 = \begin{bmatrix}
		\eta_1^T \\ \vdots \\ \eta_5^T
	\end{bmatrix} = \begin{bmatrix}
2 \sqrt{\frac{2}{89}} & -5 \sqrt{\frac{3}{178}} & 0 & \sqrt{\frac{3}{178}} & 0 & 0 & 0 & \sqrt{\frac{42}{89}} \\
-3 \sqrt{\frac{10}{979}} & -7 \sqrt{\frac{3}{9790}} & 0 & 37 \sqrt{\frac{3}{9790}} & 0 & 0 & \sqrt{\frac{178}{385}} & -8 \sqrt{\frac{6}{34265}} \\
-2 \sqrt{\frac{30}{913}} & -\frac{\sqrt{\frac{83}{110}}}{2} & 0 & -\frac{17}{2 \sqrt{9130}} & 0 & \frac{\sqrt{\frac{385}{166}}}{2} & -2 \sqrt{\frac{42}{4565}} & -\frac{1}{2} \left(17 \sqrt{\frac{7}{9130}}\right) \\
-\frac{6}{\sqrt{415}} & 0 & 0 & -\sqrt{\frac{15}{83}} & \frac{\sqrt{\frac{83}{35}}}{2} & -\frac{1}{2} \left(3 \sqrt{\frac{3}{2905}}\right) & \frac{33}{2 \sqrt{2905}} & \frac{13 \sqrt{\frac{3}{2905}}}{2} \\
-\frac{3}{2 \sqrt{10}} & 0 & \frac{\sqrt{\frac{5}{2}}}{2} & 0 & -\frac{3}{2 \sqrt{70}} & -\sqrt{\frac{3}{70}} & -\frac{3}{2 \sqrt{70}} & \sqrt{\frac{3}{70}}
\end{bmatrix}.
\end{align*} Set $R_1$ to be the first two columns of $\Gamma_2$, $R_2$ to be the third and fourth columns of $\Gamma_2$, and so on. Finally, set $\widetilde{P}_{2,v} = \frac{8}{5}R_vR_v^*$ for all $v\in [4]$, which are given by \begin{align*}
\widetilde{P}_{2,1} &=  \begin{bmatrix}\frac{364}{445} & -\frac{12}{89 \sqrt{55}} & 102 \sqrt{\frac{3}{406285}} & \frac{-96}{5} \sqrt{\frac{2}{36935}} & -\frac{24}{5 \sqrt{445}} \\
-\frac{12}{89 \sqrt{55}} & \frac{4188}{24475} & \frac{3562}{275} \sqrt{\frac{3}{7387}} & \frac{144}{5} \sqrt{\frac{2}{81257}} & \frac{36}{5 \sqrt{979}} \\
102 \sqrt{\frac{3}{406285}} & \frac{3562}{275} \sqrt{\frac{3}{7387}} & \frac{11689}{22825} & \frac{96}{415} \sqrt{\frac{6}{11}} & \frac{24}{5} \sqrt{\frac{3}{913}} \\
\frac{-96}{5} \sqrt{\frac{2}{36935}} & \frac{144}{5} \sqrt{\frac{2}{81257}} & \frac{96}{415} \sqrt{\frac{6}{11}} & \frac{288}{2075} & \frac{36}{25} \sqrt{\frac{2}{83}} \\
-\frac{24}{5 \sqrt{445}} & \frac{36}{5 \sqrt{979}} & \frac{24}{5} \sqrt{\frac{3}{913}} & \frac{36}{25} \sqrt{\frac{2}{83}} & \frac{9}{25}\end{bmatrix}, \\
\widetilde{P}_{2,2} &= \begin{bmatrix}
	 \frac{12}{445} & \frac{444}{445 \sqrt{55}} & \frac{-34}{5}\sqrt{\frac{3}{406285}} & -12 \sqrt{\frac{2}{36935}} & 0 \\
	\frac{444}{445 \sqrt{55}} & \frac{16428}{24475} & \frac{-1258}{275} \sqrt{\frac{3}{7387}} & \frac{-444}{5} \sqrt{\frac{2}{81257}} & 0 \\
	\frac{-34}{5} \sqrt{\frac{3}{406285}} & \frac{-1258}{275} \sqrt{\frac{3}{7387}} & \frac{289}{22825} & \frac{34}{415} \sqrt{\frac{6}{11}} & 0 \\
	-12 \sqrt{\frac{2}{36935}} & -\frac{444}{5} \sqrt{\frac{2}{81257}} & \frac{34}{415} \sqrt{\frac{6}{11}} & \frac{24}{83} & 0 \\
	0 & 0 & 0 & 0 & 1 \end{bmatrix}, \\
\widetilde{P}_{2,3} &= \begin{bmatrix}
	 0 & 0 & 0 & 0 & 0 \\
	0 & 0 & 0 & 0 & 0 \\
	0 & 0 & \frac{77}{83} & \frac{-3\sqrt{66}}{415}  & \frac{-2}{5} \sqrt{\frac{33}{83}} \\
	0 & 0 & \frac{-3\sqrt{66}}{415} & \frac{1976}{2075} & \frac{-33}{25} \sqrt{\frac{2}{83}} \\
	0 & 0 & \frac{-2}{5} \sqrt{\frac{33}{83}} & \frac{-33}{25} \sqrt{\frac{2}{83}} & \frac{3}{25}
\end{bmatrix}, \\
\widetilde{P}_{2,4} &= \begin{bmatrix}
	 \frac{336}{445} & -\frac{384}{445 \sqrt{55}} & \frac{-476}{5} \sqrt{\frac{3}{406285}} & \frac{156}{5} \sqrt{\frac{2}{36935}} & \frac{24}{5 \sqrt{445}} \\
	-\frac{384}{445 \sqrt{55}} & \frac{18544}{24475} & \frac{-2304}{275} \sqrt{\frac{3}{7387}} & 60 \sqrt{\frac{2}{81257}} & -\frac{36}{5 \sqrt{979}} \\
	\frac{-476}{5}  \sqrt{\frac{3}{406285}} & \frac{-2304}{275} \sqrt{\frac{3}{7387}} & \frac{3367}{22825} & \frac{-97}{415} \sqrt{\frac{6}{11}} & \frac{-2}{5} \sqrt{\frac{3}{913}} \\
	\frac{156}{5} \sqrt{\frac{2}{36935}} & 60 \sqrt{\frac{2}{81257}} & \frac{-97}{415} \sqrt{\frac{6}{11}} & \frac{456}{2075} & \frac{-3}{25}  \sqrt{\frac{2}{83}} \\
	\frac{24}{5 \sqrt{445}} & -\frac{36}{5 \sqrt{979}} & \frac{-2}{5}  \sqrt{\frac{3}{913}} & \frac{-3}{25}\sqrt{\frac{2}{83}} & \frac{3}{25}
\end{bmatrix}. \end{align*} Then, $\widetilde{P}_{2,1}, \dots, \widetilde{P}_{2,4}$ are projections such that $\widetilde{P}_{2,1} + \dots + \widetilde{P}_{2,4} = \frac{8}{5}I_5$.

\bibliographystyle{alpha}
\bibliography{self-testing-bib}

\newcommand{\etalchar}[1]{$^{#1}$}
\begin{thebibliography}{GKW{\etalchar{+}}18}

\bibitem[AMR{\etalchar{+}}19]{qiso1}
A.~Atserias, L.~Man\v{c}inska, D.~E. Roberson, R.~\v{S}\'{a}mal, S.~Severini,
  and A.~Varvitsiotis.
\newblock Quantum and non-signalling graph isomorphisms.
\newblock {\em Journal of Combinatorial Theory, Series B}, 136:289 -- 328,
  2019.

\bibitem[Arv77]{Arveson}
W.~Arveson.
\newblock Notes on extensions of {$C^{\sp*}$}-algebras.
\newblock {\em Duke Math. J.}, 44(2):329--355, 1977.

\bibitem[Bla06]{Black}
B.~Blackadar.
\newblock {\em Operator algebras}, volume 122 of {\em Encyclopaedia of
  Mathematical Sciences}.
\newblock Springer-Verlag, Berlin, 2006.
\newblock Theory of $C^*$-algebras and von Neumann algebras, Operator Algebras
  and Non-commutative Geometry, III.

\bibitem[BvCA18a]{Bow18a}
J.~Bowles, I.~\v{S}upi\'{c}, D.~Cavalcanti, and A.~Ac\'{\i}n.
\newblock Device-independent entanglement certification of all entangled
  states.
\newblock {\em Phys. Rev. Lett.}, 121:180503, Oct 2018.

\bibitem[BvCA18b]{Bow18b}
J.~Bowles, I.~\v{S}upi\'{c}, D.~Cavalcanti, and A.~Ac\'{\i}n.
\newblock Self-testing of pauli observables for device-independent entanglement
  certification.
\newblock {\em Phys. Rev. A}, 98:042336, Oct 2018.

\bibitem[CE76]{Choi-Effros}
M.~D. Choi and E.~G. Effros.
\newblock The completely positive lifting problem for {$C\sp*$}-algebras.
\newblock {\em Ann. of Math. (2)}, 104(3):585--609, 1976.

\bibitem[CGJV19]{CGJV}
A.~Coladangelo, A.~B. Grilo, S.~Jeffery, and T.~Vidick.
\newblock Verifier-on-a-leash: new schemes for verifiable delegated quantum
  computation, with quasilinear resources.
\newblock In {\em Advances in cryptology---{EUROCRYPT} 2019. {P}art {III}},
  volume 11478 of {\em Lecture Notes in Comput. Sci.}, pages 247--277.
  Springer, Cham, 2019.

\bibitem[CGS17]{Coladangelo2017}
A.~Coladangelo, K.~T. Goh, and V.~Scarani.
\newblock All pure bipartite entangled states can be self-tested.
\newblock {\em Nature Communications}, 8(1), May 2017.

\bibitem[CHSH69]{CHSH}
J.~F. Clauser, M.~A. Horne, A.~Shimony, and R.~A. Holt.
\newblock Proposed experiment to test local hidden-variable theories.
\newblock {\em Physical Review Letters}, 23(15):880--884, October 1969.

\bibitem[CHTW04]{CHTW}
R.~Cleve, P.~H{\o}yer, B.~Toner, and J.~Watrous.
\newblock Consequences and limits of nonlocal strategies.
\newblock In {\em Proceedings. 19th {IEEE} Annual Conference on Computational
  Complexity, 2004.} {IEEE}, 2004.

\bibitem[CM14]{clevemittal}
R.~Cleve and R.~Mittal.
\newblock Characterization of binary constraint system games.
\newblock In {\em Proceedings of the 41st International Colloquium on Automata,
  Languages, and Programming}, ICALP '14, pages 320--331. 2014.

\bibitem[CMMN20]{Cui2020}
D.~Cui, A.~Mehta, H.~Mousavi, and S.~S. Nezhadi.
\newblock A generalization of {CHSH} and the algebraic structure of optimal
  strategies.
\newblock {\em {Quantum}}, 4:346, October 2020.

\bibitem[Col17]{Col17}
A.~Coladangelo.
\newblock Parallel self-testing of (tilted) {EPR} pairs via copies of (tilted)
  {CHSH} and the magic square game.
\newblock {\em Quantum Inf. Comput.}, 17(9-10):831--865, 2017.

\bibitem[Col20]{C20}
A.~Coladangelo.
\newblock A two-player dimension witness based on embezzlement, and an
  elementary proof of the non-closure of the set of quantum correlations.
\newblock {\em Quantum}, 4:282, June 2020.

\bibitem[Con76]{Connes}
A.~Connes.
\newblock Classification of injective factors. {C}ases {$II_{1},$} {$II_{\infty
  },$} {$III_{\lambda },$} {$\lambda \not=1$}.
\newblock {\em Ann. of Math. (2)}, 104(1):73--115, 1976.

\bibitem[CS17]{CS17}
A.~{Coladangelo} and J.~{Stark}.
\newblock {Robust self-testing for linear constraint system games}.
\newblock {\em arXiv e-prints}, page arXiv:1709.09267, September 2017.

\bibitem[DPP19]{DPP}
K.~Dykema, V.~I. Paulsen, and J.~Prakash.
\newblock Non-closure of the set of quantum correlations via graphs.
\newblock {\em Comm. Math. Phys.}, 365(3):1125--1142, 2019.

\bibitem[EGH{\etalchar{+}}11]{EtingofRepTheory}
P.~Etingof, O.~Golberg, S.~Hensel, T.~Liu, A.~Schwendner, D.~Vaintrob, and
  E.~Yudovina.
\newblock {\em Introduction to representation theory}, volume~59 of {\em
  Student Mathematical Library}.
\newblock American Mathematical Society, Providence, RI, 2011.
\newblock With historical interludes by Slava Gerovitch.

\bibitem[FJVY19]{Fitz-et-al1}
J.~Fitzsimons, Z.~Ji, T.~Vidick, and H.~Yuen.
\newblock Quantum proof systems for iterated exponential time, and beyond.
\newblock In {\em S{TOC}'19---{P}roceedings of the 51st {A}nnual {ACM} {SIGACT}
  {S}ymposium on {T}heory of {C}omputing}, pages 473--480. ACM, New York, 2019.

\bibitem[{Fu}19]{Fu19}
H.~{Fu}.
\newblock {Constant-sized correlations are sufficient to robustly self-test
  maximally entangled states with unbounded dimension}.
\newblock {\em arXiv e-prints}, page arXiv:1911.01494, November 2019.

\bibitem[GK17]{GH17}
U.~T. Gau\`ers and O.~Khatami.
\newblock Inverse and stability theorems for approximate representations of
  finite groups.
\newblock {\em Mat. Sb.}, 208(12):70--106, 2017.

\bibitem[GKW{\etalchar{+}}18]{Goh18}
K.~T. Goh, J.~Kaniewski, E.~Wolfe, T.~V\'ertesi, X.~Wu, Y.~Cai, Y-C. Liang, and
  V.~Scarani.
\newblock Geometry of the set of quantum correlations.
\newblock {\em Phys. Rev. A}, 97:022104, Feb 2018.

\bibitem[Gow17]{Gowers17}
W.~T. Gowers.
\newblock Generalizations of {F}ourier analysis, and how to apply them.
\newblock {\em Bull. Amer. Math. Soc. (N.S.)}, 54(1):1--44, 2017.

\bibitem[HMPS19]{HMPS}
J.~W. Helton, K.~P. Meyer, V.~I. Paulsen, and M.~Satriano.
\newblock Algebras, synchronous games, and chromatic numbers of graphs.
\newblock {\em New York J. Math.}, 25:328--361, 2019.

\bibitem[HS18]{Hadwin-Shulman}
D.~Hadwin and T.~Shulman.
\newblock Tracial stability for {$C^*$}-algebras.
\newblock {\em Integral Equations Operator Theory}, 90(1):Paper No. 1, 35,
  2018.

\bibitem[JNV{\etalchar{+}}20]{JNVWY}
Z.~{Ji}, A.~{Natarajan}, T.~{Vidick}, J.~{Wright}, and H.~{Yuen}.
\newblock {MIP*=RE}.
\newblock {\em arXiv e-prints}, page arXiv:2001.04383, January 2020.

\bibitem[Kan17]{Kan17}
J.~Kaniewski.
\newblock Self-testing of binary observables based on commutation.
\newblock {\em Phys. Rev. A}, 95:062323, Jun 2017.

\bibitem[KRS02]{KRS}
S.~A. Kruglyak, V.~I. Rabanovich, and Yu.~S. Samo\u{\i}lenko.
\newblock On sums of projections.
\newblock {\em Funktsional. Anal. i Prilozhen.}, 36(3):20--35, 96, 2002.

\bibitem[Lor97]{Loring}
T.~A. Loring.
\newblock {\em Lifting solutions to perturbing problems in {$C^*$}-algebras},
  volume~8 of {\em Fields Institute Monographs}.
\newblock American Mathematical Society, Providence, RI, 1997.

\bibitem[MR14]{MR14}
L.~Man\v{c}inska and D.~Roberson.
\newblock Graph homomorphisms for quantum players.
\newblock In {\em 9th {C}onference on the {T}heory of {Q}uantum {C}omputation,
  {C}ommunication and {C}ryptography}, volume~27 of {\em LIPIcs. Leibniz Int.
  Proc. Inform.}, pages 212--216. Schloss Dagstuhl. Leibniz-Zent. Inform.,
  Wadern, 2014.

\bibitem[MR16]{qhomos}
L.~Man\v{c}inska and D.~E. Roberson.
\newblock Quantum homomorphisms.
\newblock {\em Journal of Combinatorial Theory, Series B}, 118:228 -- 267,
  2016.

\bibitem[MR20]{MR20}
M.~Musat and M.~R{\o}rdam.
\newblock Non-closure of quantum correlation matrices and factorizable channels
  that require infinite dimensional ancilla.
\newblock {\em Comm. Math. Phys.}, 375(3):1761--1776, 2020.
\newblock With an appendix by Narutaka Ozawa.

\bibitem[Mur90]{Murphy}
G.~J. Murphy.
\newblock {\em {$C^*$}-algebras and operator theory}.
\newblock Academic Press, Inc., Boston, MA, 1990.

\bibitem[MY98]{MY98}
D.~{Mayers} and A.~{Yao}.
\newblock Quantum cryptography with imperfect apparatus.
\newblock In {\em Proceedings 39th Annual Symposium on Foundations of Computer
  Science (Cat. No.98CB36280)}, pages 503--509, 1998.

\bibitem[MY04]{MY04}
D.~Mayers and A.~Yao.
\newblock Self testing quantum apparatus.
\newblock {\em Quantum Inf. Comput.}, 4(4):273--286, 2004.

\bibitem[MYS12]{MYS12}
M.~McKague, T.~H. Yang, and V.~Scarani.
\newblock Robust self-testing of the singlet.
\newblock {\em J. Phys. A}, 45(45):455304, 11, 2012.

\bibitem[NC00]{NC}
M.~A. Nielsen and I.~L. Chuang.
\newblock {\em Quantum computation and quantum information}.
\newblock Cambridge University Press, Cambridge, 2000.

\bibitem[NV17]{NV17}
A.~Natarajan and T.~Vidick.
\newblock A quantum linearity test for robustly verifying entanglement.
\newblock In {\em S{TOC}'17---{P}roceedings of the 49th {A}nnual {ACM} {SIGACT}
  {S}ymposium on {T}heory of {C}omputing}, pages 1003--1015. ACM, New York,
  2017.

\bibitem[NV18]{NV}
A.~Natarajan and T.~Vidick.
\newblock Low-degree testing for quantum states, and a quantum entangled games
  {PCP} for {QMA}.
\newblock In {\em 59th {A}nnual {IEEE} {S}ymposium on {F}oundations of
  {C}omputer {S}cience---{FOCS} 2018}, pages 731--742. IEEE Computer Soc., Los
  Alamitos, CA, 2018.

\bibitem[NW19]{NW}
A.~Natarajan and J.~Wright.
\newblock {NEEXP} is contained in {MIP}{$^\ast$}.
\newblock In {\em 2019 {IEEE} 60th Annual Symposium on Foundations of Computer
  Science ({FOCS})}. {IEEE}, November 2019.

\bibitem[Pau02]{CBMOA}
V.~Paulsen.
\newblock {\em Completely bounded maps and operator algebras}, volume~78 of
  {\em Cambridge Studies in Advanced Mathematics}.
\newblock Cambridge University Press, Cambridge, 2002.

\bibitem[Per93]{Peres}
A.~Peres.
\newblock {\em Quantum theory: concepts and methods}, volume~57 of {\em
  Fundamental Theories of Physics}.
\newblock Kluwer Academic Publishers Group, Dordrecht, 1993.

\bibitem[PSS{\etalchar{+}}16]{PSSTW}
V.~I. Paulsen, S.~Severini, D.~Stahlke, I.~G. Todorov, and A.~Winter.
\newblock Estimating quantum chromatic numbers.
\newblock {\em J. Funct. Anal.}, 270(6):2188--2222, 2016.

\bibitem[{\v{S}}B20]{SB19}
I.~{\v{S}}upi{\'{c}} and J.~Bowles.
\newblock Self-testing of quantum systems: a review.
\newblock {\em {Quantum}}, 4:337, September 2020.

\bibitem[Slo19]{Slofstra19}
W.~Slofstra.
\newblock The set of quantum correlations is not closed.
\newblock {\em Forum Math. Pi}, 7:e1, 41, 2019.

\bibitem[SSKA19]{Sarkar}
S.~{Sarkar}, D.~{Saha}, J.~{Kaniewski}, and R.~{Augusiak}.
\newblock {Self-testing quantum systems of arbitrary local dimension with
  minimal number of measurements}.
\newblock {\em arXiv e-prints}, page arXiv:1909.12722, September 2019.

\bibitem[Vid18]{Vidick}
T.~Vidick.
\newblock Quantum multiplayer games, testing and rigidity.
\newblock \url{http://users.cms.caltech.edu/~vidick/notes/ucsd/ucsd_games.pdf},
  2018.
\newblock [Online; accessed 20-February-2021].

\bibitem[Wat18]{Watr18}
J.~Watrous.
\newblock {\em The Theory of Quantum Information}.
\newblock Cambridge University Press, Cambridge, 2018.

\end{thebibliography}

\end{document}